\documentclass[12pt,onecolumn,draftclsnofoot]{IEEEtran}

\title{Linear Transmission of Composite Gaussian Measurements over a Fading Channel under Delay Constraints}

\author{\IEEEauthorblockN{Onur Tan\IEEEauthorrefmark{1}, Deniz G\"{u}nd\"{u}z\IEEEauthorrefmark{2}, Jes\'{u}s G\'{o}mez Vilardeb\`{o}\IEEEauthorrefmark{1}}\\ \normalsize
\IEEEauthorblockA{\IEEEauthorrefmark{1}Centre Tecnol\`ogic de Telecomunicacions de Catalunya (CTTC), Barcelona, Spain.}\\
\IEEEauthorblockA{\IEEEauthorrefmark{2}Imperial College London, London, UK.}

\thanks{This work was supported by the Spanish Government under project TEC2013-44591-P (INTENSYV).}
Emails: \{onur.tan@cttc.es, d.gunduz@imperial.ac.uk, jesus.gomez@cttc.es\}
}

\usepackage{amsfonts}
\usepackage{amssymb}
\usepackage{amsthm}
\usepackage{amsmath}
\usepackage{dsfont}
\usepackage{amscd}
\usepackage{graphicx, subfigure}
\usepackage{blindtext, subfigure}
\usepackage{psfrag}
\usepackage{xcolor}

\usepackage{cite}
\usepackage{epsfig}
\usepackage{epstopdf}
\usepackage{eqnarray,amsmath}
\usepackage{algorithm,algorithmicx,algpseudocode}
\usepackage{setspace}

\algnewcommand{\Inputs}[1]{%
  \State \textbf{Inputs:}
  \Statex \hspace*{\algorithmicindent}\parbox[t]{.8\linewidth}{\raggedright #1}
}
\algnewcommand{\Initialization}[1]{%
  \State \textbf{Initialization:}
  \Statex \hspace*{\algorithmicindent}\parbox[t]{.8\linewidth}{\raggedright #1}
}
\makeatletter
\def\BState{\State\hskip-\ALG@thistlm}
\makeatother

\newtheorem{thm}{Theorem}
\newtheorem{lem}{Lemma}

\begin{document}

\maketitle
\vspace{-1.7cm}

\begin{abstract}
Delay constrained linear transmission (LT) strategies are considered for the transmission of composite Gaussian measurements over an additive white Gaussian noise fading channel under an average power constraint. If the channel state information (CSI) is known by both the encoder and decoder, the optimal LT scheme in terms of the average mean-square error distortion is characterized under a strict delay constraint, and a graphical interpretation of the optimal power allocation strategy is presented. Then, for general delay constraints, two LT strategies are proposed based on the solution to a particular multiple measurements-parallel channels scenario. It is shown that the distortion decreases as the delay constraint is relaxed, and when the delay constraint is completely removed, both strategies achieve the optimal performance under certain matching conditions. If the CSI is known only by the decoder, the optimal LT strategy is derived under a strict delay constraint. The extension for general delay constraints is shown to be hard. As a first step towards understanding the structure of the optimal scheme in this case, it is shown that for the multiple measurements-parallel channels scenario, any LT scheme that uses only a one-to-one linear mapping between measurements and channels is suboptimal in general.
\end{abstract}

\vspace{-0.45cm}

\section{Introduction}
\label{s:Introduction}
Near real-time monitoring of a physical phenomena is of great significance to many emerging network applications, such as monitoring of voltage, current magnitudes, active/reactive power values in smart grids (SGs)~\cite{GridOfFuture}, or temperature and humidity in forest fire detection networks~\cite{WSNSurvey}. To this end, wireless sensors are deployed throughout the physical network and the sensor measurements are delivered to a control center (CC) over wireless links. For the robust, reliable and efficient management of the underlying physical networks, near real-time and accurate reconstruction of the measurements at the CC becomes necessary. For example, in conventional state estimation for the electricity grid, measurements are collected once every two to four seconds and the state is updated once every few minutes~\cite{PowerSystemState}. However, more frequent state measurements and estimations are required for modern SGs, which inevitably imposes strict delay constraints on the transmission of measurements. As a further example, in forest fire detection networks~\cite{WSNForestDetection}, measurements of smoke and gas sensors along with camera images are used to detect fire, and the delay inevitably becomes a major constraint for the transmission. Thus, zero-delay linear transmission (LT), rather than advanced compression and channel coding techniques that span large codewords, is an attractive strategy for the transmission of sensor measurements in intelligent networks. This is because LT reduces both the delay and encoding complexity significantly; and accordingly limits the cost and energy requirements of the sensors.

LT of Gaussian sources has been extensively studied in the literature. Goblick showed in~\cite{Goblick} that zero-delay LT of a Gaussian source over an additive white Gaussian noise (AWGN) channel achieves the optimal mean-square error (MSE) distortion. In~\cite{OptimalLinearCoding}, the optimal LT scheme that matches an $r$-dimensional Gaussian signal to a $k$-dimensional AWGN vector channel is characterized. It is shown that the optimal LT performance can be achieved by mapping ordered sources to ordered channels in a one-to-one fashion. LT of a Gaussian source over a fading AWGN channel is studied in~\cite{LinearJointSourceChannelCoding}. It is shown that the optimal LT performance can be achieved by symbol-by-symbol processing, and increasing the block length does not provide any gain, as opposed to nonlinear coding schemes. In~\cite{SmartGridStateTrans}, LT of noisy vector measurements over a fading AWGN channel is studied under diagonal and general observation matrices. LT of vector Gaussian sources over static and fading multi-antenna channels is studied in \cite{Aguerri:ISWCS:13} and~\cite{MinDistortionTranmissionBasar}, respectively.


We consider a wireless sensor node that collects measurements from $J$ Gaussian parameters. We discretize time into time slots (TSs), and assume that the CC asks for a measurement of a particular parameter from the sensor at each TS. The sensor takes one sample of the requested parameter at each TS, and transmits these samples to the CC over an AWGN fading channel under a given delay constraint. Note that, in contrast to multi-dimensional Gaussian source models studied in~\cite{OptimalLinearCoding},~\cite{MinDistortionTranmissionBasar},~\cite{LinearAnalogCoding}, where the sensor has the measurements of all the $J$ Gaussian parameters at the beginning of a TS, we assume that only one measurement is taken from the requested parameter at each TS.

We assume that each measurement must be delivered within $d$ TSs. Thereby, after each transmission, the CC estimates the measurement whose deadline is just about to expire. We assume that the channel gain from the sensor to the CC is independent and identically distributed (i.i.d.) over TSs. We consider two different scenarios regarding the channel state information (CSI)$\colon$In the first scenario, the CSI is assumed to be available to both the encoder and decoder, while in the second scenario, only the decoder has CSI. Our goal is to estimate all the requested measurements at the CC within their delay constraints with the minimum MSE distortion.

We focus explicitly on LT strategies. Assuming that the CSI is known by both the encoder and decoder, we first derive the optimal LT strategy under a strict delay constraint $(d=1)$, and show that the optimal power allocation and the corresponding distortion can be interpreted as \textit{water-filling reflected on a reciprocal mirror}. Exploiting the results of~\cite{OptimalLinearCoding}, we also derive the optimal LT strategy under a strict delay constraint for a particular scenario in which the sensor transmits the measurement vector over parallel AWGN fading channels at each TS. Then, building on our previous works~\cite{DelayConstrainedLT},~\cite{ICCPaper}, and exploiting the optimal LT strategy derived for multiple measurements-parallel channels scenario above, we propose two LT strategies for general delay constraints. In both strategies, measurements are first collected and stored in a buffer whose size depends on the delay constraint, and then, are transmitted to the CC over multiple channel accesses within the delay constraint. The two strategies consider different measurement selection criterias, which are used to select the appropriate stored measurement to be transmitted at each channel access. We then derive the theoretical lower bound (TLB) and the LT lower bound (LLB) on the achievable MSE distortion. We characterize the MSE distortion achieved by the proposed LT schemes, as well as the TLB and the LLB under various power and delay constraints. We show that the MSE distortion diminishes as the delay constraint is relaxed if the sensor is capable of measuring more than one system parameter, i.e., $J>1$. However, if $J=1$, then relaxing the delay constraint does not provide any improvement in LT performance as argued in~\cite{OptimalLinearCoding}. When the fading channel follows a discrete distribution and the delay constraint is completely removed, we show that the proposed LT strategies meet the TLB under certain matching conditions between the channel states and the paramater variances; and hence, achieve the optimal performance.

When the CSI is known only by the decoder, we first derive the optimal LT strategy under a strict delay constraint. Then, we consider the multiple measurements-parallel channels scenario under a strict delay constraint and $J>1$ assumption, and show that the optimal LT performance cannot be achieved by an LT scheme that is constrained to use only a one-to-one linear mapping between measurements and channels, as opposed to the $J=1$ case~\cite{LinearJointSourceChannelCoding}, and the CSI is known by both the encoder and decoder~\cite{OptimalLinearCoding}, respectively. Since the optimal LT strategy is elusive for $J>1$, we do not consider LT strategies for larger delay constraints. Finally, we derive the TLB on the achievable MSE distortion.

The rest of the paper is structured as follows. The system model is presented in Section~\ref{s:SystemModel}. In Sections~\ref{s:StrictDelayConstraint} to~\ref{s:Bounds} CSI is assumed at both the encoder and decoder. In Section~\ref{s:StrictDelayConstraint}, we study the optimal LT strategy under a strict delay constraint. Two LT strategies are proposed for general delay constraints in Section~\ref{s:AchievableLTStrategies}. In Section~\ref{s:Bounds}, we characterize the TLB and LLB on the achievable MSE distortion. In Section~\ref{s:NoEncoderCSI}, the optimal LT strategy is derived under a strict delay constraint along with the TLB, when the CSI is known only by the decoder. Section~\ref{s:NumericalResults} presents extensive numerical results, and finally, Section~\ref{s:Conclusions} concludes the paper.

\begin{figure}
\centering
\includegraphics[width=0.6\textwidth]{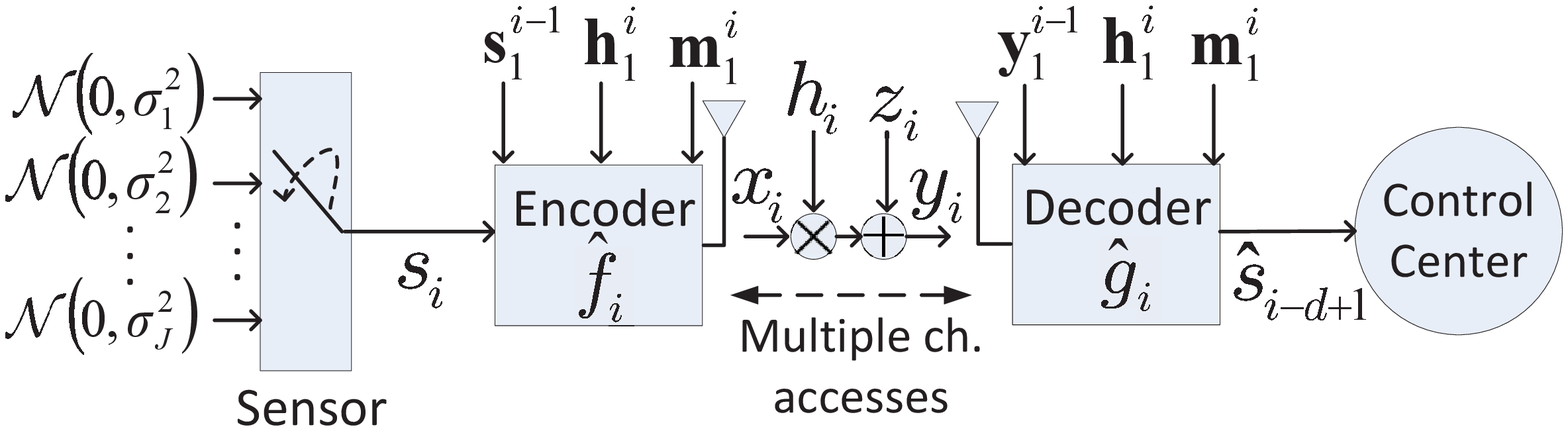}\caption{The illustration of the transmission model from the perspective of the sensor with multiple channel accesses.}%
\label{fig:TransmissionModel}%
\end{figure}

\section{System Model}
\label{s:SystemModel}
We consider a CC that monitors the operation of a system through a wireless sensor (Fig.~\ref{fig:TransmissionModel}), which is capable of measuring $J$ distinct system parameters. The $j$th system parameter is modelled as a zero-mean Gaussian random variable (r.v.) with variance $\sigma_{j}^{2}$, i.e., $\mathcal{N}(0,\sigma_{j}^{2})$, for $j\in[1\text{:}J]$, where $[1\text{:}J]$ denotes the set $\{1,2,\ldots,J\}$. These system parameters are independent from each other, and their realizations are i.i.d. over time. In order to monitor the network operation, the CC requests the measurement of one system parameter from the sensor at each TS. The index of the requested system parameter at each TS is a r.v. denoted by $M\in[1\text{:}J]$, with distribution $p_{M}(m)$, which is also i.i.d. over time. Based on these requests, the sensor takes one measurement of the requested parameter $m$ at each TS. Thereby, the model is that of a composite source introduced in Chapter $6$ of~\cite{RDTBerger}. The source $S$ can be described as a composite source comprised of $J$ distinct components (subsources), each operating independently of the others. In our model, each component produces data according to a Gaussian probability distribution $P(\cdot|m)=\mathcal{N}(0,\sigma^2_m)$. The set $G$ of all sources comprises the composite source $\{S_i,G\}$. In our case,
\begin{equation}
G=\left[\mathcal{N}(0,\sigma^2_1),\mathcal{N}(0,\sigma^2_2),\ldots,\mathcal{N}(0,\sigma^2_J)\right].
\end{equation}
The index of the requested system parameter $m$ generates the sequence of positions assumed by the switch in Fig.~\ref{fig:TransmissionModel}. In our model both the encoder and the decoder possess the exact knowledge of this sequence. Notice that, in the particular case in which the encoder and decoder are uninformed about this sequence, the composite source $\{S_i,G\}$ is equivalent to a mixture of Gaussian distributions, i.e., $P_S(s)=\sum_{m=1}^J P_M(m)P_{S|M}(s|m)$.

We assume that the CC imposes a maximum delay constraint of  $d\in\mathbb{Z}^{+}$ on the measurements, that is, the measurement requested in a TS needs to be transmitted within the following $d$ TSs; otherwise, it becomes stale. The collected sensor measurements are transmitted to the CC over a fading channel with zero-mean and unit variance AWGN. The channel output at TS $i$ is given by $y_i=h_ix_i+z_i$, where $x_i$ is the channel input, $z_i$ is the additive noise with $Z\sim\mathcal{N}(0,1)$, and $h_i$ is the fading state of the channel. We consider a fading channel model, and assume that the fading coefficient $H_i\in\mathbb{R}$ is modelled as a r.v. i.i.d. over time with probability distribution $p_{H}(h)$.

We define $\mathbf{m}_{i}^{l}=[m_{i},m_{i+1},...,m_{l}]$ as the sequence of indices of requested parameters at TSs $[i\text{:}l]$ for $i\leq l$. The measurement sequence is defined similarly as $\mathbf{s}_{i}^{l}=[s_i,\ldots, s_{l}]$, where the $i$-th entry $s_i$ is the measured value of the requested parameter $m_i$ at TS $i$. Therefore, the sequence $\mathbf{s}_{i}^{l}$ has independent entries, where the $i$-th entry comes from a Gaussian distribution with variance $\sigma_{m_i}^{2}$. Note that in our composite Gaussian measurements model, conditioned on the requested parameter index, which is known by both the encoder and decoder, the source samples follow Gaussian distributions with different variance values.

The channel state and the output sequences at TSs $[i\text{:}l]$ are similarly defined as $\mathbf{h}_{i}^{l}=[h_{i},...,h_{l}]$ and $\mathbf{y}_{i}^{l}=[y_{i},...,y_{l}]$, respectively. We assume that both the encoder and decoder at TS $i$ know all the past channel states, $\mathbf{h}_{1}^{i-1}$, and the indices of requested parameters, $\mathbf{m}_{1}^{i}$, as well as the statistics of the measured parameters, $\sigma_{m}^{2}$, the parameter requests, $p_{M}(m)$, and the channel, $p_{H}(h)$. In the first part of the paper we assume that both the encoder and decoder know the current channel state, $h_i$. Note that this assumption might be hard to realize for a fast fading channel model; on the other hand, our system model can be considered as instances of a slow fading channel. Typically, there will be a large number of sensors in the system, and each sensor is going to be scheduled only once in a while; and hence, each TS in our system model can be considered as one instance of a slow fading channel when a particular sensor is scheduled to transmit. Since these instances are separated from each other due to the transmission of other sensors, corresponding channel states are modeled as i.i.d., and are assumed to be known by both the encoder and decoder, as channel estimation and CSI feedback can be carried out between two transmissions of the same sensor. In Section~\ref{s:NoEncoderCSI} we will consider the scenario in which the CSI is known only by the decoder.

\subsubsection{Encoding Function}
The encoding function $\hat{f}_i:\mathbb{R}^{i}\times\mathbb{R}^{i}\times\mathbb{R}^{i}\rightarrow\mathbb{R}$, maps $\mathbf{s}_{1}^{i}$, $\mathbf{h}_{1}^{i}$, and $\mathbf{m}_{1}^{i}$ to a channel input $x_i$ at each TS $i$, i.e., $x_i=\hat{f}_i(\mathbf{s}_{1}^{i},\mathbf{h}_{1}^{i},\mathbf{m}_{1}^{i})$. An average power constraint of $P$ is imposed on the encoding function$\colon$
\[
\bar{P}\triangleq\lim_{n\to\infty}\frac{1}{n}\sum\limits_{i=1}^{n}\mathrm{E}_{M,H,S}\left[|X_i|^2\right]
\leq P,
\]
\noindent where $\mathrm{E}_{M,H,S}\left[\cdot\right]$ denotes the expectation over $M$, $H$ and $S$. For any generic transmission policy, the encoding function $\hat{f}_i$, at TS $i$, may consider to use any combination of $\mathbf{s}_{1}^{i}$, $\mathbf{h}_{1}^{i}$, and $\mathbf{m}_{1}^{i}$ to generate $x_i$. This gives rise to a time-varying encoding scheme.

\subsubsection{Decoding Function}
At the end of TS $i$, the goal of the CC is to estimate with the minimum MSE distortion, the measurement $s_{i-d+1}$, which has been requested exactly $d-1$ TSs ago, and is about to expire. The decoding function $\hat{g}_i:\mathbb{R}^{i}\times\mathbb{R}^{i}\times\mathbb{R}^{i}\rightarrow\mathbb{R}$, estimates $\hat{s}_{i-d+1}$ based on $\mathbf{y}_{1}^{i}$, $\mathbf{h}_{1}^{i}$, and $\mathbf{m}_{1}^{i}$, i.e., $\hat{s}_{i-d+1}=\hat{g}_i(\mathbf{y}_{1}^{i},\mathbf{h}_{1}^{i},\mathbf{m}_{1}^{i})$. The MSE distortion is given by$\colon$

 \[
 \bar{D}\triangleq\lim_{n\to\infty}\frac{1}{n}\sum\limits_{i=d}^{n}\mathrm{E}_{M,H,S,Z}\left[|S_{i-d+1}-\hat
 {S}_{i-d+1}|^2\right].
 \]

The decoding function $\hat{g}_i$, at TS $i$, reconstructs the measurement using $\mathbf{y}_{1}^{i}$, $\mathbf{h}_{1}^{i}$, and $\mathbf{m}_{1}^{i}$. Hence, similarly to the encoder, the decoder may be time-varying.

We are interested only in LT policies in which $\hat{f}_i$'s are restricted to be linear functions of the sensor measurements, $s_i$'s, i.e., $\hat{f}_i(\mathbf{s}_{1}^{i},\mathbf{h}_{1}^{i},\mathbf{m}_{1}^{i})\triangleq f_i(\mathbf{h}_{1}^{i},\mathbf{m}_{1}^{i})\cdot\mathbf{s}_{1}^{i}$, where $f_i(\mathbf{h}_{1}^{i},\mathbf{m}_{1}^{i})$ is a vector. Under this linearity constraint, the optimal estimators at the receiver, $\hat{g}_i$'s, are also linear functions of the channel outputs, $y_i$'s, i.e., $\hat{g}_i(\mathbf{y}_{1}^{i},\mathbf{h}_{1}^{i},\mathbf{m}_{1}^{i})\triangleq g_i(\mathbf{h}_{1}^{i},\mathbf{m}_{1}^{i})\cdot\mathbf{y}_{1}^{i}$, where $g_i(\mathbf{h}_{1}^{i},\mathbf{m}_{1}^{i})$ is a vector. Hereafter, we will refer to $f_i$ and $g_i$ for the encoding and decoding functions at TS $i$, respectively.

\section{Strict Delay Constraint}
\label{s:StrictDelayConstraint}
We first derive the optimal LT strategy under a strict delay constraint $(d=1)$, and characterize the minimum achievable MSE distortion. In this scenario, optimal LT performance is achieved by transmitting only the most recent measurement since all the previous measurements have expired, and transmitting an expired measurement cannot help the estimation of the current measurement since the measurements are independent. Then the encoding function $f_i(h_i,m_i)$ at TS $i$ is a scalar. Given the encoding function, the decoding function $g_i(h_i,m_i)$ that minimizes the MSE for Gaussian r.v.s is the linear MMSE estimator~\cite{NetworkInformationTheory}, and is also a scalar.

In particular, for a measurement $s_i$ with variance $\sigma_{m_i}^2$, and channel output $y_i=h_i \cdot f_i(h_i,m_i) \cdot s_i+z_i$ at TS $i$, the decoding function can be written explicitly as$\colon$
\vspace{-0.05cm}
\begin{align}
\label{eq:optimallineardecoder-CSI}
g_i(h_i,m_i)=\frac{\mathrm{E}_{S,Z}[S_iY_i]}{\mathrm{E}_{S,Z}[Y_i^2]}=\frac{|h_i|f_i(h_i,m_i)\sigma_{m_i}^2}{|h_i|^2f_i(h_i,m_i)^2\sigma_{m_i}^2+1}.
\end{align}

In the following lemma we show that there is no loss of optimality by limiting the encoding function to be time-invariant.

\begin{lem}
\label{Lemma1}
Under a strict delay constraint there is no loss of optimality by considering only time-invariant encoding functions, i.e., $f_i(h_i,m_i)=f(h_i,m_i)$ $\forall i$.
\end{lem}

\begin{proof}
\vspace{-0.05cm}
\begin{align}
\label{eq:TimeInvariant-1}
\hspace{-0.1cm}\bar{D}&=\lim_{n\to\infty}\frac{1}{n}\sum\limits_{i=1}^{n}\mathrm{E}_{M,H,S,Z}\left[|S_{i}-\hat
 {S}_{i}|^2\right]=\lim_{n\to\infty}\frac{1}{n}\sum\limits_{i=1}^{n}\mathrm{E}_{M,H}\left[\frac{\sigma_{m}^2}{|h|^2f_i(h,m)^2\sigma_{m}^2+1}\right],\\
\label{eq:TimeInvariant-2}
&\ge\mathrm{E}_{M,H}\left[\frac{\sigma_{m}^2}{|h|^2f(h,m)^2\sigma_{m}^2+1}\right],
\end{align}

\noindent where~(\ref{eq:TimeInvariant-1}) is the average MSE distortion under a strict delay constraint $(d=1)$; and defining $f(h,m)^2\triangleq\lim_{n\to\infty}\frac{1}{n}\sum\limits_{i=1}^{n}f_i(h,m)^2$ such that $f(h,m)$ satisfies the average power constraint $P$,~(\ref{eq:TimeInvariant-2}) follows from the convexity of the function $\mathrm{E}_{M,H}\left[\frac{\sigma_{m}^2}{|h|^2f_i(h,m)^2\sigma_{m}^2+1}\right]$ in terms of $f_i(h,m)^2$, and the equality holds iff $f_i(h,m)=f(h,m)$ for $\forall i$ and due to the strict convexity of the aforementioned function. Thus, the time-invariant encoding function $f(h,m)$, which is a function of only $h$ and $\sigma_{m}^2$, does not lead to any loss in optimality.
\end{proof}

The time-invariant encoding function $f(h,m)$ leads to a time-invariant decoding function $g(h,m)$. In the rest of the paper, we will consider time-invariant encoding and decoding functions without loss of optimality. Then, the MSE distortion, $\bar{D}=\mathrm{E}_{M,H,S,Z}[|S-\hat{S}|^2]$, and the average power, $\bar{P}=\mathrm{E}_{M,H,S}[|X|^2]$, can be written explicitly as functions of $h$ and $\sigma_{m}^{2}$, as follows$\colon$
\begin{align}
\label{eq:s-PAoverfadingchannelandsources}
&\bar{D}=\sum\limits_{m=1}^{J}p_M(m)\int_\mathbb{R}\frac{\sigma_{m}^2}{|h|^2f(h,m)^2\sigma_{m}^2+1}p_H(h)\mathrm{d}h,\\
&\bar{P}=\sum\limits_{m=1}^{J}p_M(m)\int_\mathbb{R}f(h,m)^2\sigma_{m}^2p_H(h)\mathrm{d}h.
\end{align}

The optimal linear encoding function $f^{\ast}(h,m)$ is found as the solution to the convex optimization problem $\bar{D}^{\ast}\triangleq\underset{f}{\text{min}}\,\bar{D}$, subject to the average power constraint $\bar{P}\leq P$. From the Karush-Kuhn-Tucker (KKT) optimality conditions~\cite{ConvexOptimization}, we obtain$\colon$

\vspace{-0.05cm}
\begin{align}
f^{\ast}(h,m)=\sqrt{\left[ \frac{\lambda^{\ast}}{|h|\sigma_{m}}-\frac{1}{|h|^2\sigma_{m}^2}\right]  ^{+}},
\end{align}

\noindent where $\lambda^{\ast}$ is the optimal Lagrange multiplier, such that $\bar{P}=P$.

The optimal power allocation and the corresponding distortion are given by$\colon$
\begin{align}
\label{eq:SingleChannelCase_OptPowerAllocation}
P^{\ast}(h,m)&=\frac{\sigma_m}{|h|}\left[\lambda^{\ast}-\frac{1}{|h|\sigma_m}\right]^{+},\\
\label{eq:SingleChannelCase_CorrespondingDistortion}
D^{\ast}(h,m)&=\frac{\sigma_m}{|h|}\min\left(\frac{1}{\lambda^{\ast}},|h|\sigma_m\right),
\end{align}

\noindent where $\bar{D}^{\ast}=\mathrm{E}_{M,H}\left[D^{\ast}(h,m)\right]$ and $\mathrm{E}_{M,H}\left[P^{\ast}(h,m)\right]=P$.

\begin{figure}
\centering
\includegraphics[width=0.6\textwidth]{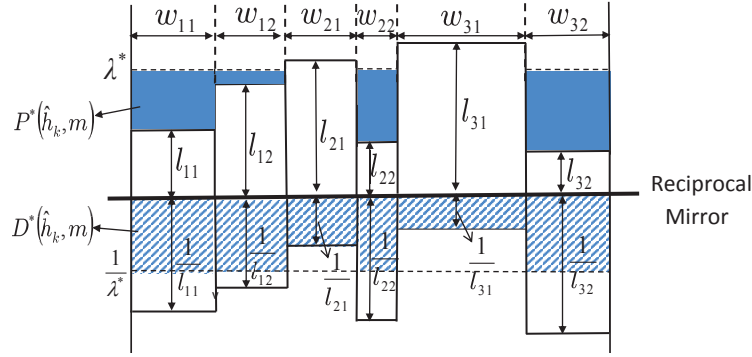}
\caption{Water-filling reflected on a reciprocal mirror.}
\label{fig:ReflectionMirror}
\end{figure}

In Fig.~\ref{fig:ReflectionMirror}, we present a graphical interpretation of the optimal power allocation and the corresponding distortion for $J=2$ parameters with variances $\sigma_{1}^{2}$ and $\sigma_{2}^{2}$, which are requested with probabilities $p_{M}(1),p_{M}(2)$, respectively. We also consider a discrete fading channel with three states, where the $k$th state, $\hat{h}_{k}$, is observed with probability $p_{H}(\hat{h}_{k})$, $k=1,2,3$. Fig.~\ref{fig:ReflectionMirror} depicts rectangles that are placed upon a mirror surface and their reciprocally scaled images below. Rectangles represent all possible source-channel pairs $\{\sigma_m,\hat{h}_k\}$, where $l_{km}\triangleq \frac{1}{|\hat{h}_{k}|\sigma_{m}}$ and $w_{km}\triangleq \frac{\sigma_{m}}{|\hat{h}_{k}|}$ indicate the height and width of the rectangles, respectively. The total power is poured above the level $l_{km}$ up to the water level $\lambda^{\ast}$ across the rectangles placed upon the mirror. The optimal power allocated to the source-channel pair $\{\sigma_m,\hat{h}_k\}$ is given by the shaded area below the water level and above $l_{km}$. The corresponding distortion values are found by simply looking at the reciprocally scaled reflections of the rectangles and the water level on the mirror. If $\frac{1}{l_{km}}>\frac{1}{\lambda^{\ast}}$, distortion is given by the width $w_{km}$ times the reciprocal of the water level $\frac{1}{\lambda^{\ast}}$, and if $\frac{1}{l_{km}}\le\frac{1}{\lambda^{\ast}}$, distortion is $\sigma_m^2$, which are illustrated as dashed areas in Fig.~\ref{fig:ReflectionMirror}. We call this as \textit{water-filling reflected on a reciprocal mirror}.

\vspace{-0.05cm}
\subsection{Multiple Measurements and Parallel Channels}
\label{s:ParallelChannelsCase}
Next, we assume that the CC requests $N>1$ measurements from the sensor at each TS, and the sensor transmits a length-$N$ measurement vector over $N$ parallel orthogonal AWGN fading channels under a strict delay constraint $(d=1)$. For this scenario, we characterize the optimal LT strategy by generalizing the results of~\cite{OptimalLinearCoding} derived for Gaussian vector sources to our composite Gaussian measurements model. This scenario differs from the system model defined in Section~\ref{s:SystemModel}, since we allow to take $N$ measurements at each TS as opposed to taking only one measurement at each TS. However, we will exploit the optimal LT strategy in this setting for the construction of the proposed transmission strategies in Section~\ref{s:AchievableLTStrategies}, as well as for characterizing the LLB in Section~\ref{s:LinearLowerBound}.

Only for this scenario, we define $\mathbf{m}=\left[m_1,...,m_N \right]$ as the vector of indices of $N$ requested parameters at a particular TS. Then, the sensor takes the length-$N$ measurement vector $\mathbf{s}=\left[s_{1}, \ldots, s_{N}\right]$ according to the parameters indicated by $\mathbf{m}$, i.e., $s_1$ is the measured value of parameter $m_1$. For a strict delay constraint $(d=1)$, the optimal LT performance is achieved by transmitting only the most recent measurement vector. Similarly to Lemma~\ref{Lemma1}, the encoding function can be limited to a time-invariant $N\times N$ square matrix $\mathbf{F_{h,m}}$ without loss of optimality, where subscripts $\mathbf{h}$ and $\mathbf{m}$ indicate the dependence of the encoding matrix on the realizations of $\mathbf{h}$ and $\mathbf{m}$. We assume that the $N$ channels are i.i.d with distribution $p_H(h)$, and $\mathbf{H}$ is defined as the $N\times N$ diagonal channel matrix. The diagonal elements of $\mathbf{H}$ are denoted by a channel vector $\mathbf{h}=\left[h_{1},\ldots,h_{N}\right]$ at a particular TS. The length-$N$ channel output vector at that particular TS is given by $\mathbf{y}=\mathbf{Hx}+\mathbf{z}$, where $\mathbf{x}$ is the length-$N$ channel input vector and $\mathbf{z}$ is the length-$N$ additive noise vector with $\mathbf{z}\sim\mathcal{N}(\mathbf{0,I})$.

The encoder at any TS maps its measurement vector $\mathbf{s}$, to a channel input vector $\mathbf{x}$, i.e., $\mathbf{x}=\mathbf{F_{h,m}}\mathbf{s}$. An average power constraint of $P$ is imposed on the encoding function$\colon$
\vspace{-0.05cm}
\begin{align}
\label{eq:AveragePowerParallelChannels}
\bar{P}=\frac{1}{N}\operatorname{Tr}{\big\{\mathrm{E}_{M,H,S}[\mathbf{xx}^{T}]\big\}}=\frac{1}{N}\operatorname{Tr}{\big\{\mathrm{E}_{M,H}[\mathbf{ F_{h,m}}\mathbf{C_{s}}{\mathbf{F}}^T_{\mathbf{h,m}}]\big\}}
 \leq P,
\end{align}

\noindent where $\mathbf{C_{s}}=\mathrm{E}_{S}[\mathbf {s}\mathbf {s}^T]$.

Given the encoding function, the decoding function that minimizes the MSE for a Gaussian random vector is the $N\times N$ linear MMSE estimator matrix $\mathbf{G_{h,m}}$~\cite{NetworkInformationTheory}, which is also time-invariant. Similarly to $\mathbf{F_{h,m}}$, subscripts $\mathbf{h}$ and $\mathbf{m}$ indicate the dependence of the decoding matrix on the realizations of $\mathbf{h}$ and $\mathbf{m}$. For the measurement vector $\mathbf{s}$, and the channel output vector $\mathbf{y}$, at any TS, we have$\colon$
\vspace{-0.2cm}
\begin{align}\label{eq:OptimalLinearDecoderRelaxedDelay}
\mathbf{G_{h,m}}=\mathbf{C_{sy}}\mathbf{C_{y}^{-1}}=\mathbf{C_{s}}{\mathbf{F}}^T_{\mathbf{h,m}}\mathbf{H}^T\boldsymbol\Phi,
\end{align}

\noindent where $\mathbf{C_{sy}}=\mathrm{E}_{S,Z}[\mathbf {s}\mathbf {y}^T]$, $\mathbf{C_{y}}=\mathrm{E}_{S,Z}[\mathbf {y}\mathbf {y}^T]$ and $\boldsymbol\Phi\triangleq (\mathbf{H}{\mathbf{F_{h,m}}}\mathbf{C_{s}}{\mathbf{F}}^T_{\mathbf{h,m}}\mathbf{H}^T+\mathbf{I})^{-1}$.

At any TS, the CC estimates the most recent measurement vector $\mathbf{s}$ as $\mathbf{\hat{s}}$, i.e., $\mathbf{\hat{s}}=\mathbf{G_{h,m}}\mathbf{y}$. The MSE distortion is given by$\colon$

\begin{align}
\label{eq:AverageDistortionParallelChannels}
 \bar{D}=\frac{1}{N}\operatorname{Tr}{\big\{\mathrm{E}_{M,H,S,Z}[\left|\mathbf{s}-\mathbf{\hat
 {s}}||\mathbf{s}-\mathbf{\hat{s}}|^{T}\right]\big\}}=\frac{1}{N}\operatorname{Tr}{\big\{\mathrm{E}_{M,H}[\mathbf{C_{s}}-\mathbf{C_{s}}{\mathbf{F}}^T_{\mathbf{h,m}}\mathbf{H}^T\boldsymbol\Phi\mathbf{H}{\mathbf{F_{h,m}}}\mathbf{C_{s}}]\big\}}.
\end{align}

The optimal linear encoding matrix $\mathbf{F}^{\ast}_{\mathbf{h,m}}$, is found as the solution to the convex optimization problem $\bar{D}^{\ast}\triangleq\underset{\mathbf{F_{h,m}}}{\text{min }}\bar{D}$, subject to the average power constraint $\bar{P}\leq P$. For a set of static parallel AWGN channels and Gaussian vector sources, the optimal linear encoding matrix transmits one measurement over each channel~\cite{OptimalLinearCoding}. The optimal mapping between channels and measurements is as follows: We first reorder the measurement vector $\mathbf{s}$ to obtain $\mathbf{\bar{s}}=[s_{(1)},\ldots,s_{(N)}]$, such that $\sigma_{m_{(1)}}^{2}\leq\sigma_{m_{(2)}}^{2}\leq\cdots\leq\sigma_{m_{(N)}}^{2}$, and reorder the channel vector $\mathbf{h}$ to obtain $\mathbf{\bar{h}}=\left[{h}_{(1)},\ldots,{h}_{(N)}\right]$, such that $|{h}_{(1)}|\leq|{h}_{(2)}|\leq\cdots\leq|{h}_{(N)}|$. Then, the optimal linear encoding matrix $\mathbf{F}^{\ast}_{\mathbf{h,m}}$ is diagonal with entries $\left[f_{(1)}(h_{(1)},m_{(1)}),\ldots,f_{(N)}(h_{(N)},m_{(N)})\right]$, and it maps the ordered measurements to ordered channel states. In order to find the diagonal entries of $\mathbf{F}^{\ast}_{\mathbf{h,m}}$, we can explicitly rewrite the convex optimization problem by using the optimal mappings derived in~\cite{OptimalLinearCoding}, as follows$\colon$

\begin{equation}
\begin{aligned}
\label{opt:ParallelChannelsOptProblem}
\bar{D}^{\ast}\triangleq &\underset{f_{(t)}}{\text{ min }} \mathrm{E}_{M_{(t)},H_{(t)}}\Bigg[\frac{1}{N}\sum\limits_{t=1}^{N}\frac{\sigma_{m_{(t)}}^2}{|{h}_{(t)}|^2{f}_{(t)}(h_{(t)},m_{(t)})^2\sigma_{m_{(t)}}^2+1}\Bigg]\\
&\text{s.t. }
\mathrm{E}_{M_{(t)},H_{(t)}}\Bigg[\frac{1}{N}\sum\limits_{t=1}^{N}{f}_{(t)}(h_{(t)},m_{(t)})^2\sigma_{m_{(t)}}^2\Bigg] \le P,
\end{aligned}
\end{equation}

\noindent where the expectation is taken over $M_{(t)}$ and $H_{(t)}$ for $t\in[1\text{:}N]$. The $t$-th smallest entry of the requested parameter vector $\mathbf{m}=[m_{1},m_{2},\ldots,m_{N}]$, is denoted by the r.v. $M_{(t)}\in[1\text{:}J]$ with the order statistics $p_{M_{(t)}}(m)$. Without loss of generality, we assume that ordering the entries of $\mathbf{m}$ in ascending order, i.e., $m_{(1)}\leq m_{(2)}\leq\cdots\leq m_{(N)}$, implies ordering the entries of the measurement vector $\mathbf{s}$ in ascending variances, i.e., $\sigma_{m_{(1)}}^{2}\leq\sigma_{m_{(2)}}^{2}\leq\cdots\leq\sigma_{m_{(N)}}^{2}$. Similarly, the $t$-th smallest entry of the channel vector $\mathbf{h}=[h_{1},h_{2},\ldots,h_{N}]$ is denoted by the r.v. $H_{(t)}\in\mathbb{R}$ with the order statistics $p_{H_{(t)}}(h)$.

The optimal linear encoding matrix $\mathbf{F}^{\ast}_{\mathbf{h,m}}$ with diagonal entries ${f}_{(t)}^{\ast}(h_{(t)},m_{(t)})$ for $t\in[1\text{:}N]$, can be found from the Lagrange and the KKT conditions as follows$\colon$

\begin{align}
\label{OptimalEncoderRelaxed}
{f}_{(t)}^{\ast}(h_{(t)},m_{(t)})=\sqrt{\left[\frac{\delta^{\ast}}{|{h}_{(t)}|{\sigma}_{m_{(t)}}}-\frac{1}{|{h}_{(t)}|^2{\sigma}_{m_{(t)}}^2}\right]^{+}},
\end{align}

\noindent where $\delta^{\ast}$ is the optimal Lagrange multiplier, such that $\bar{P}=P$ in~(\ref{opt:ParallelChannelsOptProblem}).

Similarly, the optimal power allocation and the corresponding distortion can be found by using the \textit{water-filling reflected on a reciprocal mirror} interpretation. The optimal Lagrange multiplier $\delta^{\ast}$ depends on $p_{M_{(t)}}(m)$ and $p_{H_{(t)}}(h)$, which can be found explicitly by using the order statistics. In the following lemma, we give the $t$-th order statistics $p_{M_{(t)}}(m)$ and $p_{H_{(t)}}(h)$, for $t\in[1\text{:}N]$.

\begin{lem}
\label{Lemma2}
Let $F_M(m)$ and $F_H(h)$ denote the cumulative distribution functions of $p_M(m)$ and $p_H(h)$, respectively. Given $F_M(m)$, $p_M(m)$, $F_H(h)$, $p_H(h)$ and $N$, the $t$-th order statistics $p_{M_{(t)}}(m)$ and $p_{H_{(t)}}(h)$, $t\in[1\text{:}N]$, are found as$\colon$


\begin{align}
\label{eq:DensityFunctions}
&p_{H_{(t)}}(h)=tp_H(h)\binom{N}{t}(F_H(h))^{t-1}(1-F_H(h))^{N-t},\\
&p_{M_{(t)}}(m)=\sum\limits_{b=t}^{N}\binom{N}{b}\left[F_M(m)^b(1-F_M(m))^{N-b}-F_M(m-1)^b(1-F_M(m-1))^{N-b}\right].
\end{align}

\end{lem}

\begin{proof}
The proof is trivial and achieved through the definition of the cumulative distribution functions of $H_{(t)}$ and $M_{(t)}$.

\begin{align}
\label{eq:OrderDistribution1}
F_{H_{(t)}}(h)&=\mathrm{Pr}\{H_{(t)} \le h\}=\mathrm{Pr}\{\text{at least $t$ of $H$'s are}\le h\},\\
\label{eq:OrderDistribution2}
&=\sum\limits_{b=t}^{N}\frac{N!}{(N-b)!b!}F_H(h)^b(1-F_H(h))^{N-b}.
\end{align}
\noindent where~(\ref{eq:OrderDistribution1}) implies a binomial distribution with at least $t$ successes and can be formulated as~(\ref{eq:OrderDistribution2}). The $t$-th order statistics $p_{H_{(t)}}(h)$ is found by taking the derivative of~(\ref{eq:OrderDistribution2}) with respect to $h$. The same proof holds for $M_{(t)}$.
\end{proof}

\section{LT Strategies}
\label{s:AchievableLTStrategies}
In this section, we propose two LT strategies for general delay constraints $d\geq1$. The block diagram of the proposed LT strategies is illustrated in Fig.~\ref{fig:LTHM-BlockDiag}. Both strategies are composed of two main blocks, namely, storage and transmission blocks. There are two buffers of size $\bar{d}$ measurements, namely, the measurement buffer (MB) and the transmission buffer (TB). Here, we present these two schemes for an odd delay constraint, i.e., $d\in\{1,3,5,\ldots\}$, but they can be easily adapted to the case when $d$ is even. In the storage block, given a delay constraint of $d=2\bar{d}-1$ for $\bar{d}\in[1\text{:}\infty]$, the sensor collects a block of $\bar{d}$ consecutive measurements after $\bar{d}$ consecutive TSs, and stores them in the MB. The consecutive blocks of $\bar{d}$ measurements, taken over successive time intervals, are indexed by $\bar{k}=\{1,2,\ldots\}$. Then, the $\bar{k}$-th block consists of the measurements taken within TSs $[(1+(\bar{k}-1)\bar{d})\text{:}\bar{k}\bar{d}]$, i.e., $\mathbf{s}_{(1+(\bar{k}-1)\bar{d})}^{\bar{k}\bar{d}}$. When the MB gets full with the $\bar{d}$ measurements of the $\bar{k}$-th block, the sensor removes $\mathbf{s}_{(1+(\bar{k}-1)\bar{d})}^{\bar{k}\bar{d}}$ from the MB  and loads them into the TB. Then, for the next consecutive $\bar{d}$ TSs $[\bar{k}\bar{d}\text{:}((\bar{k}+1)\bar{d}-1)]$, the sensor accesses the channel and transmits a linear function of the measurements in the TB, i.e., $\mathbf{s}_{(1+(\bar{k}-1)\bar{d})}^{\bar{k}\bar{d}}$, over the channel states $\mathbf{h}_{\bar{k}\bar{d}}^{((\bar{k}+1)\bar{d}-1)}$ satisfying the delay constraint $d$. The specifics of these linear functions will be explained below.

\begin{figure}
\centering
\includegraphics[width=0.6\textwidth]{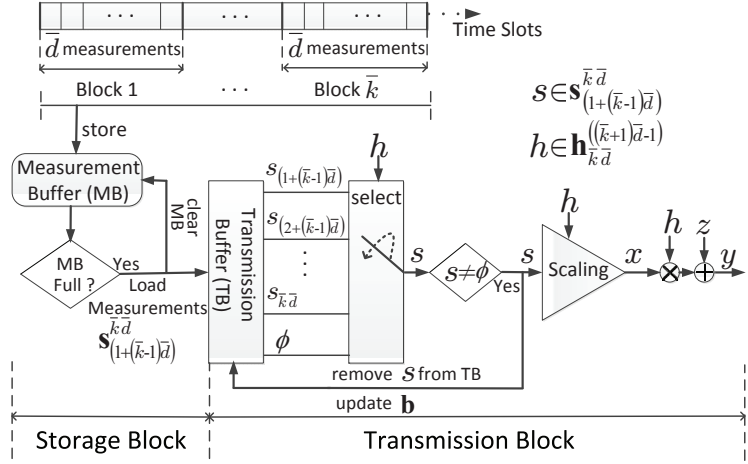}\caption{The block diagram illustration of the proposed LT strategies.}%
\label{fig:LTHM-BlockDiag}
\end{figure}

Note that, while the sensor transmits the measurements in the TB, it starts refilling the MB with new measurements $\mathbf{s}_{(\bar{k}\bar{d}+1)}^{(\bar{k}\bar{d}+\bar{d})}$. After $\bar{d}$ channel accesses within TSs $[\bar{k}\bar{d}\text{:}((\bar{k}+1)\bar{d}-1)]$, the MB gets full again and its new $\bar{d}$ measurements are transferred to the TB for transmission over the next $\bar{d}$ TSs.

The proposed transmission strategies consist of two sub-blocks, namely, the measurement selection and scaling sub-blocks. This division is motivated by the results of~\cite{OptimalLinearCoding} presented in Section~\ref{s:ParallelChannelsCase}, in which $N$ ordered measurements are mapped one-to-one to $N$ ordered channels, and each measurement is transmitted over its corresponding channel. Hence, we assume that, at each channel access, the sensor selects only one measurement and scales it to a channel input value. However, in this case, we cannot directly use the optimal LT scheme in~\cite{OptimalLinearCoding} and guarantee that the selected measurement and the channel state satisfy the optimal matching. This is because even though $\bar{d}$ measurements are available in the TB in advance, the states of the next $\bar{d}$ channels are not available to the transmitter as in the parallel channel model of~\cite{OptimalLinearCoding}; and instead, they become available over time. The two proposed LT strategies differ in the way they choose the measurement to be transmitted at each TS.

\subsection{Linear Transmission Scheme with Hard Matching (LTHM)}
\label{s:LTHM}
This transmission scheme has the following measurement selection criteria. Assume, without loss of generality, that parameters are ordered such that $\sigma_1^2>\sigma_2^2>\cdots>\sigma_J^2$. We divide the channel magnitude space ($\mathbb{R}^+$) into $J$ ordered channel intervals as, $\mathcal{H}_m=[{H_{m}^{\prime}},{H_{(m-1)}^{\prime}})$, where ${H_{m}^{\prime}}<{H_{(m-1)}^{\prime}}$ for $m\in[1\text{:}J]$. The boundary values are chosen as ${H_{0}^{\prime}}=\infty$, ${H_{J}^{\prime}}=0$ and ${H_{m}^{\prime}}=F_{H}^{-1}(1-\sum\limits_{j=1}^{m}p_M(j))$, for $m\in[1\text{:}(J-1)]$, where $F_{H}^{-1}(\cdot)$ denotes the inverse of the cumulative distribution function of the channel magnitude $|h|$, $F_{H}(|h|)$. Observe that according to this choice, the probability of the channel magnitude belonging to $\mathcal{H}_m$ is $\mathrm{Pr}\{|h|\in\mathcal{H}_m\}=p_{M}(m)$.\footnote{If channel fading follows a discrete distribution, we define sets of channel states as opposed to intervals. With abuse of notation, we denote the $m$th set as $\mathcal{H}_m$, for $m\in[1\text{:}J]$. Suppose that the discrete channel states are ordered as $|\hat{h}_1|>|\hat{h}_2|>|\hat{h}_3|>\cdots$. We allocate the discrete states into $J$ sets such that the probability of channel state falling into set $\mathcal{H}_m$ is $p_M(m)$. However, it may be possible that the channel states cannot be grouped to satisfy this equality exactly for all $m$. In that case we create virtual states to satisfy these equalities, as explained below.

Let $j$ be the minimum index for which $\sum_{i=1}^j p_H(|h|= \hat{h}_i) > p_M(1)$. Define $p_M^1 = p_M(1) - \sum_{i=1}^{j-1} p_H(|h|= \hat{h}_i)$. We define a new virtual channel state $\hat{h}_j^1$, whose gain is equivalent to $\hat{h}_j$. Whenever the real channel state is $\hat{h}_j$, we randomly assign the channel state to $\hat{h}_j^1$ with probability $p_M^1/p_M(j)$. We let $\mathcal{H}_1 = \{\hat{h}_1, \ldots, \hat{h}_{j-1}, \hat{h}_j^1\}$. We repeat the same process for $p_M(2)$, starting with channel state $\hat{h}_j$ whose probability is now $p_H(\hat{h}_j) - p_M^1$.}

The algorithmic description of LTHM is given in Algorithm~\ref{LTSchemes}. Let $\mathbf{b}=[b_{1}, b_{2},\ldots, b_{J}]$ be a $J$-length vector, where the $m$-th entry, $b_{m}\in[0\text{:}\bar{d}]$, denotes the number of measurements of parameter $m$ in the TB, for $m\in[1\text{:}J]$. At each channel access, if $|h|\in\mathcal{H}_m$ and $b_m\neq0$, then the sensor selects one measurement of the parameter type $m$ from the TB and feeds it to the scaling sub-block. If there are multiple measurements of the same parameter type $m$ in the TB, i.e., $b_m>1$, then the sensor selects one of them randomly. The selected measurement is removed from the TB and $\mathbf{b}$ is updated by reducing the $m$-th entry, $b_m$, by one. Thereby, each measurement is transmitted only once. On the other hand, if $|h|\in\mathcal{H}_m$ and $b_m=0$, no measurement is transmitted in that TS. Hence, LTHM considers a hard matching condition for selecting measurements, in which each parameter has a corresponding interval of channel states, and only measurements of that parameter can be transmitted over a channel state from that interval. Note that, since the channel state is known at the receiver, it also knows which type of measurement is transmitted at each TS.

For the scaling sub-block we use the power allocation strategy derived in Section~\ref{s:StrictDelayConstraint}. Thus, the selected measurement of the parameter type $m$ is transmitted at the current channel state $|h|\in\mathcal{H}_m$, for $m\in[1\text{:}J]$, by allocating power $P(h,m)$, leading to distortion $D(h,m)$$\colon$
\begin{align}
\label{eq:PowerAllocationLTHM}
\hspace{-0.2cm}  P(h,m)=
\begin{cases}
    \Big[\frac{\mu\sigma_m}{|h|}-\frac{1}{|h|^2}\Big]^+, &\text{if } \text{hard matching holds,}\\
    0,              & \text{otherwise.}
\end{cases}
\end{align}

\begin{align}
\label{eq:DistortionLTHM}
\hspace{-0.6cm} D(h,m)=
\begin{cases}
    \frac{\sigma_{m}^2}{|h|^2P(h,m)+1}, &\text{if } \text{hard matching holds,}\\
    \sigma_{m}^2,              & \text{otherwise,}
\end{cases}
\end{align}

\noindent where $\mu$ is chosen such that the average power constraint is satisfied.

After every transmission, the CC estimates the transmitted measurement $s$ by using the channel output $y$. It is noteworthy that after $\bar{d}$ channel accesses, we may have untransmitted measurements in the TB. TB is emptied anyway since these measurements have expired, and they are estimated with the maximum distortion $\sigma^2_m$. As we show next, the average number of untransmitted measurements decreases with the increasing delay constraint $d$. However, for a finite delay constraint the untransmitted measurements dominate the distortion even for a high average transmission power constraint. In order to combat this drawback, we propose an alternative LT scheme.

\begin{spacing}{1}
\begin{algorithm}
  \caption{LTHM and LTSM}
  \label{LTSchemes}
  \begin{algorithmic}[1]
   \Initialization{Load measurements of MB, $\mathbf{s}_{(1+(\bar{k}-1)\bar{d})}^{\bar{k}\bar{d}}$, into TB and update $\mathbf{b}$.}
   \For{$i$ = $\bar{k}\bar{d}$ to $(\bar{k}+1)\bar{d}-1$}\Comment{TSs for $\bar{d}$ channel accesses}
   \If{$|h_i|\in\mathcal{H}_m$ and $b_m\neq0$}\Comment{both for LTHM and LTSM}
   \Statex \text{Select one measurement of parameter $m$ from TB.} \Comment{measurement selection}
   \Statex \text{Transmit the measurement over $|h_i|$ with an allocated power of Eqn.~(\ref{eq:PowerAllocationLTHM}).}\Comment{scaling}
   \Statex\text{$b_m \gets b_m-1$}\Comment{update $\mathbf{b}$}
   \Else\Comment{only for LTSM}
   \Statex \text{Find $\varsigma$ by solving $\underset{b_\varsigma\neq0}{\text{min }}\left\lvert |h_i|-{h_{\varsigma}^{\prime}} \right\rvert.$}
   \Statex \text{Select one measurement of parameter $\varsigma$ from TB.}\Comment{measurement selection}
   \Statex \text{Transmit the measurement over $|h_i|$ with an allocated power of Eqn.~(\ref{eq:PowerAllocationLTHM}).}\Comment{scaling}
   \Statex\text{$b_m \gets b_m-1$}\Comment{update $\mathbf{b}$}
   \EndIf
   \EndFor
   \Statex \text{$\bar{k}\gets\bar{k}+1$ and go to \textbf{Initialization}}
   \end{algorithmic}
\end{algorithm}
\end{spacing}

\subsection{Linear Transmission Scheme with Soft Matching (LTSM)}
\label{s:LTSM}
The algorithmic description of LTSM is given in Algorithm~\ref{LTSchemes}. The LTSM retains the hard matching condition of LTHM, i.e., at each channel access, if $|h|\in\mathcal{H}_m$ and $b_m\neq0$ for $m\in[1\text{:}J]$, LTSM selects one measurement of the parameter type $m$ from the TB. Hence, LTSM also gives the highest selection priority to the measurement of the parameter type that satisfies the hard matching condition with the channel state. However, if $|h|\in\mathcal{H}_m$ and $b_m=0$, LTSM does not waste the channel state; and instead, selects one measurement based on the following measurement selection criteria$\colon$

Assume that each interval $\mathcal{H}_m$ is further divided into two equally probable intervals by the boundary value ${h_{m}^{\prime}}=F_{H}^{-1}\left(\frac{F_{H}(H_{(m-1)}^{\prime})+F_{H}(H_{m}^{\prime})}{2}\right)$, for $\forall m\in[1\text{:}J]$\footnote{If the channel follows a discrete fading distribution, we find $h_{m}^{\prime}$ by taking the mean value of all elements of channel set $\mathcal{H}_m$.}. If $|h|\in\mathcal{H}_m$ and $b_m=0$, then LTSM selects one measurement of parameter $\varsigma$, which is the parameter that minimizes the following distance metric$\colon$

\begin{equation}
\begin{aligned}
\label{eq:NotPerfectMatching}
\underset{b_\varsigma\neq0}{\text{min }}\left\lvert |h|-{h_{\varsigma}^{\prime}} \right\rvert.
\end{aligned}
\end{equation}

When the hard matching condition is not satisfied, the LTSM considers a soft matching condition for selecting measurements; that is, among all parameter types of the measurements in the TB, it selects a measurement of the parameter whose corresponding interval of channel states has the value ${h_{\varsigma}^{\prime}}$ closest to the channel state magnitude $|h|$. If two distinct $\varsigma$ values satisfy the solution of Eqn.~(\ref{eq:NotPerfectMatching}), then LTSM chooses the smallest value of $\varsigma$.  LTSM allocates the power as in Eqn.~(\ref{eq:PowerAllocationLTHM}), and transmits the selected measurement, leading to distortion in Eqn.~(\ref{eq:DistortionLTHM}). Note that the optimal Lagrange multiplier $\mu$ is chosen such that the average power constraint is satisfied. At the end of $\bar{d}$ channel accesses, the sensor will have transmitted all the measurements in the TB, albeit some might have been allocated zero power as a result of the water-filling algorithm.

\section{Distortion Lower Bounds}
\label{s:Bounds}
We characterize two lower bounds on the MSE distortion, namely, the TLB and the LLB. While the TLB is the theoretical performance bound derived without any delay or complexity constraints on the transmission, the LLB is a performance lower bound only for LT strategies. We also prove that the proposed LT strategies meet the TLB under infinite delay and certain matching conditions between the channel states and parameter variances.

\subsection{Theoretical Lower Bound (TLB)}
\label{s:TheoreticalLowerBound}
Shannon's source-channel separation theorem states that the optimal end-to-end distortion is achieved by concatenating the optimal source and channel codes when there is no delay or complexity constraints, and the source and channel distributions are ergodic~\cite{ElementsInformationTheory}. When we remove the delay and linear encoding constraints in our system model, then the sensor can transmit to the CC at the ergodic capacity, $\bar{C}_{e}$, of the underlying fading channel, while the minimum distortion, $\bar{D}_{e}$, is found by evaluating the distortion-rate function for a composite Gaussian source model at the ergodic capacity.

Since the channel state is known by both the transmitter and receiver, the ergodic capacity, in terms of the optimal power allocation scheme $P_e^*(h)$, is given by$\colon$
\vspace{-0.05cm}
\begin{align}
 \label{eq:ChannelCapacity}
 \bar{C}_e\triangleq\mathrm{E}_{H}\left[\frac{1}{2}\log\left(1+|h|^2P_e^*(h)\right)\right],
\end{align}

\noindent where $P_e^*(h)$ is found by the water-filling algorithm as $P_e^*(h)= [\alpha^*- 1/|h|^2]^+$, where $\alpha^*$ is chosen to satisfy $\bar{P}_e\triangleq\mathrm{E}_{H} \left[P_e^*(h) \right]=P$.


From Eqn. (6.1.21) of~\cite{RDTBerger}, the distortion-rate function of a composite Gaussian source with $m$ components, $\mathcal{N}(0,\sigma_{m}^2)$, each of which is observed with probability $p_M(m)$ for $m\in[1\text{:}J]$, is defined as$\colon$
\vspace{-0.05cm}
\begin{align}
\label{eq:DistortionRate}
\bar{D}_e\triangleq\mathrm{E}_{M}\left[\sigma_{m}^22^{-2R_e^*(\sigma_{m})}\right],
\end{align}
\noindent where the optimal rate allocated to source $m$, $R_e^*(\sigma_{m})$, and the corresponding distortion, $D_e^*(\sigma_{m})$, are given by$\colon$
\vspace{-0.1cm}
\begin{align}
\label{eq:ShannonRateAllocation}
R_e^*(\sigma_{m})&=\frac{1}{2}\left[\log\left(\frac{\sigma_{m}^2}{\beta^*}\right)\right]^+,
\end{align}
\begin{align}
\label{eq:ShannonCorrespondingDistortion}
D_e^*(\sigma_{m})&=\min\left(\beta^*,\sigma_{m}^2\right),
\end{align}
\noindent where $\beta^*$ is chosen such that $\bar{R}_{e}\triangleq\mathrm{E}_{M}\left[R_e^*(\sigma_{m})\right]=\bar{C}_e$.

Hence, the optimal distortion is found as $\bar{D}_e=\mathrm{E}_{M}\left[D_e^*(\sigma_{m})\right]$, which is the TLB on the achievable MSE distortion by any transmission strategy. Note that we have removed both the delay constraint and the linearity requirement on the encoder and decoder.

\subsubsection{Asymptotic Optimality of LT}
\label{s:AsymptoticOptimalityOfLT}
In general, the TLB cannot be achieved by LT strategies even if the delay constraint is removed. However, it can be shown that LTHM and LTSM meet this lower bound when the delay constraint is removed under certain matching conditions between the channel states and the parameter variances.

Assume that the channel follows a discrete fading distribution, where the channel state $h$ can take one of the $J$ values $\hat{h}_{m}$ with probability $p_{H}(\hat{h}_{m})$ for $m\in[1\text{:}J]$. The discrete values are ordered as $|\hat{h}_1|>|\hat{h}_2|>\cdots>|\hat{h}_J|$. The next theorem states the necessary conditions in this discrete channel model under which LTHM and LTSM achieve the optimal distortion performance when the delay constraint is removed.

\begin{thm}
\label{Theorem1}
For the discrete AWGN fading channel model, if the parameter variances and the discrete channel states satisfy
$\frac{\sigma_{1}}{|\hat{h}_{1}|}=\cdots=\frac{\sigma_{J}}{|\hat{h}_{J}|}$, and $p_{M}(m)=p_{H}(\hat{h}_{m})$, for
$\forall m\in[1\text{:}J]$, then the TLB is achieved by LTHM and LTSM when the delay constraint is removed, i.e., $d\to\infty$.
\end{thm}

\begin{proof}
The proof can be found in Appendix A.
\end{proof}

\subsection{The Linear Transmission Lower Bound (LLB)} \label{s:LinearLowerBound}

We next derive a lower bound on the achievable MSE distortion as a function of the delay and power constraints for any LT strategy. In order to derive this lower bound, we relax the assumption on the causal knowledge of the measurements and channel states, and instead assume that the sensor has the offline (non-causal) knowledge of a certain number of future measurements and channel states.  Accordingly, we assume that at any TS the sensor non-causally knows the length-$\bar{u}$ measurement vector, i.e., $\mathbf{s}=[s_{1},\ldots,s_{\bar{u}}]$, taken over the next $\bar{u}$ TSs. Observe that, for a delay constraint $d$, each measurement of $\mathbf{s}$ can only be transmitted over the following $d$ channel states observed after it is taken, thus the transmission of the vector $\mathbf{s}$ spans the following $\bar{c}=(d+\bar{u}-1)$ channel states observed after the first measurement $s_{1}$ is taken. We further assume that the sensor non-causally knows the length-$\bar{c}$ channel vector $\mathbf{h}=\left[{h}_{1},\ldots,{h}_{\bar{c}}\right]$. Henceforth, the problem is reduced to optimally transmitting $\bar{u}$ measurements over $\bar{c}$ parallel channels, which is attained by using the optimal LT scheme presented in Section~\ref{s:ParallelChannelsCase}. Accordingly, we first reorder $\mathbf{s}$ to get $\mathbf{\bar{s}}=[s_{(1)},\ldots,s_{(\bar{u})}]$, where the variances of the ordered measurements satisfy $\sigma_{m_{(1)}}^{2}\leq\sigma_{m_{(2)}}^{2}\leq\cdots\leq \sigma_{m_{(\bar{u})}}^{2}$, and reorder $\mathbf{h}$ to get $\mathbf{\bar{h}}=\left[{h}_{(1)},\ldots,{h}_{(\bar{c})}\right]$, such that the ordered fading states satisfy $|{h}_{(1)}|\leq|{h}_{(2)}|\leq\ldots\leq|{h}_{(\bar{c})}|$. Then, the $\bar{c}\times \bar{u}$ optimal linear encoding matrix $\mathbf{F}^{\ast}_{\mathbf{h,m}}$ consists of a $\bar{u}\times \bar{u}$ size diagonal partition with entries $\left[f_{(1)}({h}_{(1+\bar{e})},{m}_{(1)}),\ldots,f_{(\bar{u})}({h}_{(\bar{u}+\bar{e})},{m}_{(\bar{u})})\right]$, and a $\bar{e}\times \bar{u}$ size partition with zero entries, where $\bar{e}=\bar{c}-\bar{u}$, and it maps $\bar{u}$ ordered measurements to the $\bar{u}$ channels with the largest gains. The optimal entries of $\mathbf{F}^{\ast}_{\mathbf{h,m}}$ are found as the solution of the following convex optimization problem with the optimal objective function $\bar{D}^{\ast}(d,\bar{u},P)$$\colon$

\small
\begin{equation}
\begin{aligned}
\label{opt:LLBOptProblem}
\underset{ f_{(t)}}{\text{ min } }& \mathrm{E}_{M_{(t)},H_{(t+\bar{e})}}\Bigg[\frac{1}{\bar{u}}\sum\limits_{t=1}^{\bar{u}}\frac{\sigma_{m_{(t)}}^2}{|{h}_{(t+\bar{e})}|^2{f}_{(t)}({h}_{(t+\bar{e})},{m}_{(t)})^2\sigma_{m_{(t)}}^2+1}\Bigg]\\
\text{s.t. }
&\bar{P}\triangleq\mathrm{E}_{M_{(t)},H_{(t+\bar{e})}}\Bigg[\frac{1}{\bar{u}}\sum\limits_{t=1}^{\bar{u}}{f}_{(t)}({h}_{(t+\bar{e})},{m}_{(t)})^2\sigma_{m_{(t)}}^2\Bigg] \le P,
\end{aligned}
\end{equation}
\normalsize

\noindent where the expectation is taken over $M_{(t)}$ and $H_{(t+\bar{e})}$ for $t\in[1\text{:}\bar{u}]$. The $t$-th and $(t+\bar{e})$-th order statistics $p_{M_{(t)}}(m)$ and $p_{H_{(t+\bar{e})}}(h)$, are given by Lemma~\ref{Lemma2}. The optimal linear encoding matrix with diagonal entries is found as$\colon$

\begin{align}
\label{OptimalEncoderLB}
{f}_{(t)}^{\ast}(h_{(t+\bar{e})},m_{(t)})=\sqrt{\left[\frac{\zeta^{\ast}}{|{h}_{(t+\bar{e})}|{\sigma}_{m_{(t)}}}-\frac{1}{|{h}_{(t+\bar{e})}|^2{\sigma}_{m_{(t)}}^2}\right]^{+}},
\end{align}

\noindent where $\zeta^{\ast}$ is the optimal Lagrange multiplier, such that $\bar{P}=P$ in~(\ref{opt:LLBOptProblem}).

Assuming non-causal knowledge of $\bar{u}$ measurements and $\bar{c}$ channel states under the delay constraint $d$ and the average power constraint $P$, we obtain the optimal distortion $\bar{D}^{\ast}(d,\bar{u},P)$ for any LT strategy. Then, the LLB is derived by finding the $\bar{u}$ value, which maximizes $\bar{D}^{\ast}(d,\bar{u},P)$$\colon$

\begin{equation}
\begin{aligned}
\label{opt:Opt3}
\bar{D}_l(d,P)&\triangleq\underset{\bar{u}}{\text{max } }\bar{D}^{\ast}(d,\bar{u},P).
\end{aligned}
\end{equation}

Note that we have relaxed the constraint for the causal knowledge of measurements and channel states both at the encoder and decoder. The numerical comparisons of the LLB with the proposed schemes will be presented in Section~\ref{s:NumericalResults}.

\section{No CSI at the Encoder}
\label{s:NoEncoderCSI}
In this section, we assume that the CSI is known only at the decoder. We derive the optimal LT strategy under a strict delay constraint $(d=1)$, as well as the TLB on the achievable MSE distortion. Additionally, for the multiple measurements-parallel channels scenario studied in Section~\ref{s:ParallelChannelsCase}, we show that if the CSI is available only at the receiver, any LT scheme that is limited to a one-to-one linear mapping from the measurements to the channel input is suboptimal in general. The optimal LT strategy is elusive and it will be a non-trivial function of the source variances and the channel distribution.

\subsection{Strict Delay Constraint}
\label{s:s:StrictDelayConstraintNoCSI}
Under a strict delay constraint, the most recent measurement is transmitted at each TS. By applying Lemma~\ref{Lemma1} to this scenario, we can similarly show that there is no loss of optimality by considering time-invariant encoding functions, i.e., $f_i(m)=f(m)$, $\forall i$. Hence, the encoding function $f(m)$ is a scalar and time-invariant. The decoding function $g(h,m)$ that minimizes the MSE is the linear MMSE estimator~\cite{NetworkInformationTheory}, and is also a scalar and time-invariant. Then, the MSE distortion, $\bar{D}=\mathrm{E}_{M,H,S,Z}[|S-\hat
 {S}|^2]$, and the average power, $\bar{P}=\mathrm{E}_{M,S}[|X|^2]$, can be written explicitly as$\colon$
\vspace{-0.1cm}
\begin{align}
\label{eq:NOCSI_Distortion}
&\bar{D}=\sum\limits_{m=1}^{J}p_M(m)\int_\mathbb{R}\frac{\sigma_{m}^2}{|h|^2f(m)^2\sigma_{m}^2+1}p_H(h)\mathrm{d}h,\\
\label{eq:NOCSI_Power}
&\bar{P}=\sum\limits_{m=1}^{J}p_M(m) f(m)^2\sigma_{m}^2.
\end{align}

\noindent where $P(m)\triangleq f(m)^2\sigma_{m}^2$. The optimal linear encoding function, $f^{\ast}(m)$, is found as the solution to the convex optimization problem $\bar{D}^{\ast}\triangleq\underset{f}{\text{min}}\,\bar{D}$, subject to the average power constraint $\bar{P}\leq P$. From the KKT conditions~\cite{ConvexOptimization}, we have$\colon$

\begin{align}
f^{\ast}(m)=\sqrt{\frac{\left[\Psi^{-1}(\frac{\lambda^{\ast}}{\sigma_{m}^2})\right]^{+}}{\sigma_{m}^2}},
\end{align}

\noindent where $\Psi^{-1}:\mathbb{R}\rightarrow\mathbb{R}$ is the inverse of the function $\Psi:\mathbb{R}\rightarrow\mathbb{R}$, that is defined as, $\Psi(P(m))\triangleq\int_\mathbb{R}\frac{|h|^2}{(|h|^2P(m)+1)^2}p_H(h)\mathrm{d}h$. The optimal Lagrange multiplier $\lambda^{\ast}$ is chosen such that $\bar{P}=P$ in~(\ref{eq:NOCSI_Power}).

\subsection{Multiple Measurements and Parallel Channels}
\label{s:ParallelChannelsCaseNoCSI}
Next we consider the multiple measurements-parallel channels scenario studied in Section~\ref{s:ParallelChannelsCase}, under the strict delay constraint and the assumption that the CSI is known only at the decoder, and $J>1$. In such a scenario, the optimal LT scheme of~\cite{OptimalLinearCoding}, in which the ordered measurements are mapped one-to-one to ordered channel states, cannot be used directly. This is because, even though the encoder knows the $N$ measurements, it does not know any of the channel states, and hence; cannot order them. For the special case where $N$ measurements are observed from a single Gaussian source $(J=1)$, in~\cite{LinearJointSourceChannelCoding} the authors show that the optimal performance is achieved by transmitting one measurement over each channel. When $J=1$, since $N$ measurements all have the same variance, all orderings are equivalent, and the optimal LT performance is achieved by an LT scheme that uses only a one-to-one mapping between measurements and channels. However, this is not the case in general when $J>1$. Since $N$ measurements follow a composite Gaussian source model, the encoder can have measurements with different variances; and hence, we can exploit the diversity of the fading channel by transmitting a single measurement over multiple channels, instead of transmitting each measurement only once. Depending on the source variances, the former may surpass the best LT performance achieved by using only a one-to-one linear mapping. This is shown in the following lemma by considering a particular example.

\begin{lem}
\label{Lemma3}
Consider the LT of $N$ measurements of a composite Gaussian source with $J>1$ components over $N$ parallel AWGN fading channels. If the CSI is known only by the decoder, then the LT scheme that uses a one-to-one linear mapping between measurements and channels is suboptimal in general.
\end{lem}

\begin{proof}
The proof can be found in Appendix B.
\end{proof}

\subsection{The Theoretical Lower Bound (TLB)}
\label{s:TheoreticalLowerBoundNoCSI}
Similarly to Section~\ref{s:TheoreticalLowerBound}, we derive the TLB on the achievable MSE distortion by using Shannon's source-channel separation theorem. If the CSI is available only at the decoder and the average power constraint is $P$, then the ergodic capacity is given by$\colon$

\begin{align}
 \label{eq:ChannelCapacityNoCSI}
 \bar{C}_e\triangleq\mathrm{E}_{H}\left[\frac{1}{2}\log\left(1+|h|^2P\right)\right].
\end{align}

The distortion-rate function of a composite Gaussian source is defined as in Eqn.~(\ref{eq:DistortionRate}) of Section~\ref{s:TheoreticalLowerBound}, which leads to the optimal rate allocated to source $m$, $R_e^*(\sigma_{m})$, as in Eqn.~(\ref{eq:ShannonRateAllocation}) and the corresponding distortion, $D_e^*(\sigma_{m})$, as in Eqn.~(\ref{eq:ShannonCorrespondingDistortion}), respectively. The Lagrangian multiplier $\beta^*$ for this case is chosen such that $\mathrm{E}_{M}\left[R_e^*(\sigma_{m})\right]$ is equal to the ergodic capacity $\bar{C}_e$ in~(\ref{eq:ChannelCapacityNoCSI}). Then the TLB on the achievable MSE distortion by any strategy when the encoder does not have the CSI is given by $\bar{D}_e=\mathrm{E}_{M}\left[D_e^*(\sigma_{m})\right]$.

\section{Numerical Results and Observations}
\label{s:NumericalResults}
Here we provide numerical results to compare the performances of LTHM and LTSM with the lower bounds, and to analyze the impact of the delay and power constraints on the performance. In our simulations, we consider $J=4$ Gaussian parameters with variances $\{10, 5, 1, 0.5\}$, which are requested with probabilities $\{0.1, 0.3, 0.4, 0.2\}$, respectively. For a continuous fading channel, we consider Rayleigh distribution with a scale parameter $\omega=3$, and for a discrete fading channel, we consider four states $\{\sqrt{10},\sqrt{5},1,\sqrt{0.5}\}$, which are observed with probabilities $\{0.1, 0.3, 0.4, 0.2\}$, respectively.

\begin{figure}
\centering
\includegraphics[width=0.57\textwidth]{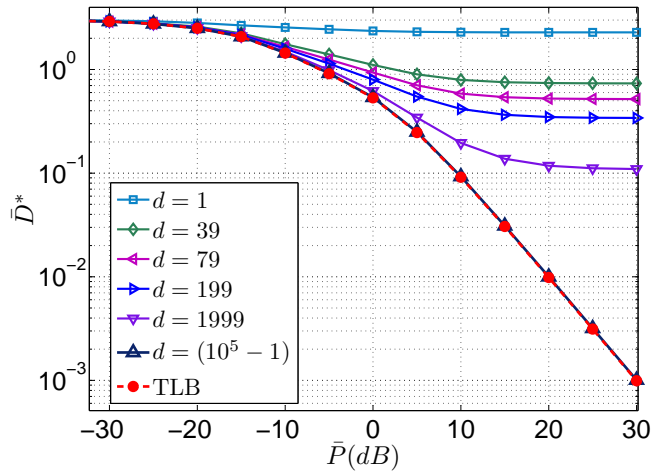}
\caption{Achievable MSE distortion with LTHM with respect to average power for different delay constraints in the discrete fading channel model.}
\label{fig:LTHM_disc}
\end{figure}

We illustrate the achievable MSE distortion versus average power under various delay constraints with LTHM in the discrete channel setting in Fig.~\ref{fig:LTHM_disc}. We observe that the MSE distortion diminishes as the delay constraint is relaxed. This is because a relaxed delay constraint provides a larger number of measurements in the TB; and hence, more flexibility for the sensor in selecting the appropriate measurement for each TS. We note that this statement does not hold when $J=1$, in which case increasing the block length does not provide any improvement~\cite{OptimalLinearCoding}. As it can be seen in Fig.~\ref{fig:LTHM_disc}, the MSE distortion converges to a fixed value as the average power value increases. This is due to the additional distortion brought in by the untransmitted measurements in the TB. The average number of untransmitted measurements and their effect on the MSE distortion decreases as the delay constraint is relaxed, since having a larger number of measurements in the TB increases the probability of finding a measurement that satisfies the hard matching condition. In particular, when the delay constraint is removed, as seen in Fig.~\ref{fig:LTHM_disc}, LTHM achieves the TLB, and becomes the optimal LT scheme, since the source-channel matching conditions in Theorem~\ref{Theorem1} are satisfied for the setup considered here.

In Fig.~\ref{fig:LTSM_cont}, we illustrate the achievable MSE distortion with LTSM with respect to average power under various delay constraints in the continuous channel model. Similarly to LTHM, the MSE distortion diminishes as the delay constraint increases. On the other hand, as opposed to LTHM, the MSE distortion achieved by LTSM decreases monotonically with the average power as illustrated in Fig.~\ref{fig:LTSM_cont}. This is because the performance of LTSM does not suffer from a fixed distortion component due to the untransmitted measurements. In addition, LTSM also approaches the TLB as the delay constraint is relaxed. Although we do not expect the LTSM to meet the TLB in this setting since the matching conditions of Theorem~\ref{Theorem1} do not hold, we observe in Fig.~\ref{fig:LTSM_cont} that it is very close to the TLB.

\begin{figure}
\centering
\includegraphics[width=0.57\textwidth]{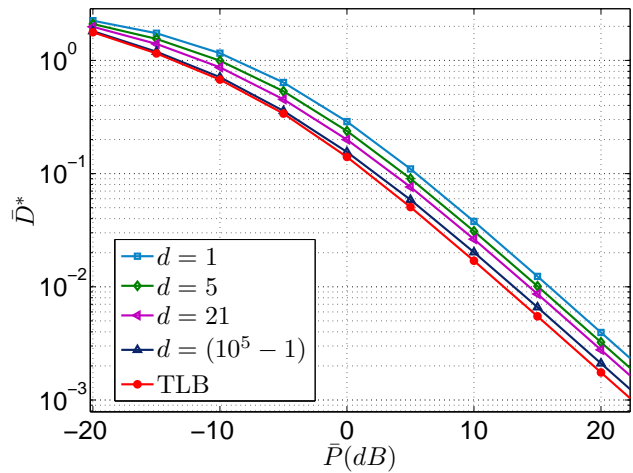}
\caption{Achievable MSE distortion with LTSM with respect to average power for various delay constraints in the continuous fading channel model.}
\label{fig:LTSM_cont}
\end{figure}

\begin{figure}
\centering
\includegraphics[width=0.57\textwidth]{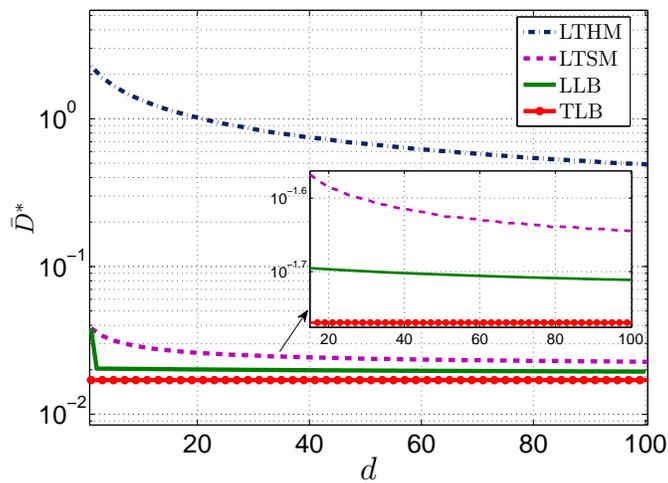}
\caption{MSE distortion versus delay constraint, $d$, in the continuous fading channel model for an average power constraint $\bar{P}=10$ dB.}
\label{fig:DelayPerformanceELT}
\end{figure}

\begin{figure}
\centering
\includegraphics[width=0.57\textwidth]{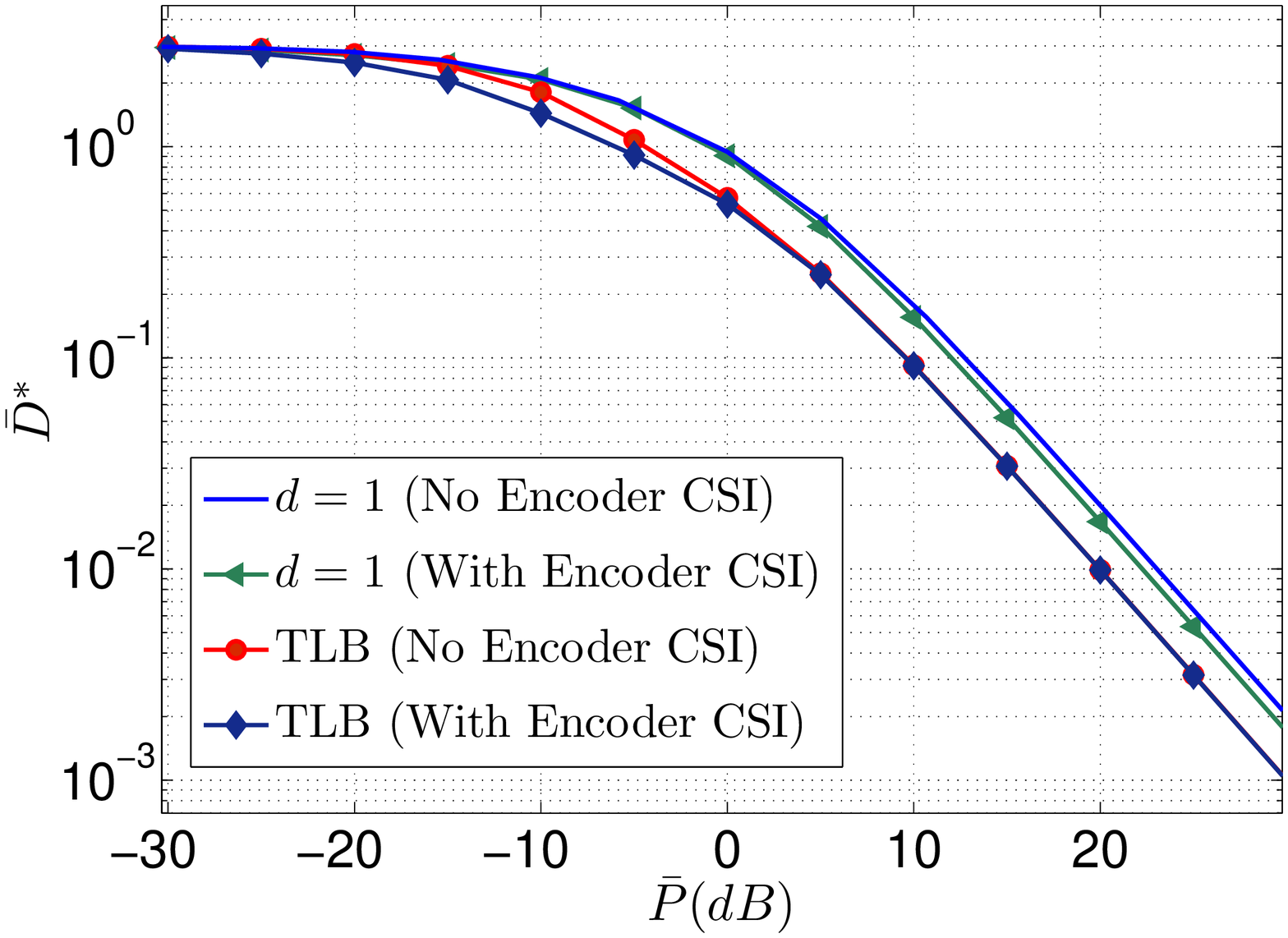}
\caption{The achievable MSE distortion of LT and the TLB with respect to average power in the discrete fading channel model with and without encoder CSI.}
\label{fig:StrictDelayCompareNoCSI}
\end{figure}

Next, we compare the performances of LTHM and LTSM with each other and with the TLB and the LLB. Fig.~\ref{fig:DelayPerformanceELT} shows the achievable MSE distortion of LTHM, LTSM, the LLB and the TLB with respect to delay constraint in the continuous fading channel model for an average power constraint $\bar{P}=10$ dB. As seen in the figure, the performance of the TLB is constant since it is derived by completely removing the delay and complexity constraints. On the other hand, the LLB decays slowly as the delay constraint increases. As expected, the MSE distortion of LTHM and LTSM decrease as the delay constraint increases. We can see that LTSM meets the LLB under the strict delay constraint. As expected, LTSM always outperforms LTHM, while the gap between the two schemes decreases with the increasing delay constraint. The gap between the TLB and two schemes also decreases with the increasing delay constraint even though we do not expect either of the schemes converge to the TLB in this setting since the matching conditions of Theorem~\ref{Theorem1} do not hold.

Finally, in Fig.~\ref{fig:StrictDelayCompareNoCSI}, we illustrate the achievable MSE distortion of LT and the TLB with respect to average power in the discrete channel model for the scenarios in which the CSI is known only by the decoder, and by both the encoder and decoder. The MSE distortion of LT under strict delay constraint of $d=1$ for both scenarios diminishes as the average power increases. However, there is a constant gap between the optimal performances achieved with and without encoder CSI at higher $\bar{P}$ values. On the other hand, the TLB for both scenarios meet as the average power increases since the gain from the optimal power allocation over different channel states disappears in the high power regime.

\section{Conclusions}
\label{s:Conclusions}
We have studied the delay-constrained LT of composite Gaussian measurements from a sensor to a CC over an AWGN fading channel. We have considered a wireless sensor that can collect measurements from $J$ distinct Gaussian parameters. The CC asks for a measurement of a particular parameter from the sensor with a certain probability at each TS. In this framework, we have presented the optimal LT strategy under a strict delay constraint, and have given a graphical interpretation for the optimal power allocation scheme and the corresponding distortion value. Then, we have proposed two LT strategies, called LTHM and LTSM, under general delay constraints, and have provided numerical results to investigate the impact of the delay and average power constraints on the performance. We have seen that, if the number of parameters, $J$, is more than one, the MSE distortion decreases as the delay constraint is relaxed. We have also derived lower bounds on the achievable MSE distortion for generic and LT strategies. While LTSM outperforms LTHM at all delay constraints, we have shown analytically that both strategies meet the lower bound when the delay constraint is removed, under certain matching conditions between the parameter and the channel statistics.

We have also studied the scenario in which the CSI is known only by the decoder. We have presented the optimal LT strategy under a strict delay constraint. We have derived a TLB on the achievable MSE distortion by relaxing the delay constraint and the linearity requirement. We have also considered the multiple measurements-parallel channels scenario under a strict delay constraint, and have shown that the optimal LT performance cannot be achieved by using only a one-to-one linear mapping between measurements and channels, as opposed to the results derived in~\cite{OptimalLinearCoding} and~\cite{LinearJointSourceChannelCoding}. The design of the optimal LT strategy for the multiple measurements-parallel channels scenario for arbitrary delay constraints is elusive, and is left as future work.

\appendices

\section{Proof of Theorem 1}
Given a delay constraint $d=2\bar{d}-1$, let the r.v. $\bar{Z}_m$, $m\in[1\text{:}J]$, denote the total number of measurements of parameter $m$ among $\bar{d}$ measurements loaded into the TB. $\bar{Z}_m$ follows a Binomial distribution with parameters $\bar{d}$ and $p_M(m)$. Hence, the probability of having $\bar{k}$ measurements of parameter $m$ in the TB is given by $p_{\bar{Z}_m}(\bar{k})=\mathrm{Pr}\{\bar{Z}_m=\bar{k}\}=\binom{\bar{d}}{\bar{k}}p_M(m)^{\bar{k}}(1-p_M(m))^{\bar{d}-\bar{k}}$. Similarly, considering the discrete fading model presented in Section~\ref{s:AsymptoticOptimalityOfLT}, let the r.v. $\hat{Z}_m$, $m\in[1\text{:}J]$, denote the total number of channels with state $\hat{h}_m$, after $\bar{d}$ channel accesses. $\hat{Z}_m$ also follows a Binomial distribution with parameters $\bar{d}$ and $p_H(\hat{h}_m)$. Hence, the probability of observing $\hat{k}$ channels with state $\hat{h}_m$ is given by $p_{\hat{Z}_m}(\hat{k})=\mathrm{Pr}\{\hat{Z}_m=\hat{k}\}=\binom{\bar{d}}{\hat{k}}p_H(\hat{h}_m)^{\hat{k}}(1-p_H(\hat{h}_m))^{\bar{d}-\hat{k}}$.

Observe that after $\bar{d}$ channel accesses, the number of transmitted measurements selected from the TB with LTHM is given by $\text{min}\{\bar{Z}_m,\hat{Z}_m\}$. On the other hand, the number of untransmitted measurements remained in the TB is given by $[\bar{Z}_m-\hat{Z}_m]^{+}$. Then, the average power, $\bar{P}_{\infty}$, and the achievable MSE distortion, $\bar{D}_{\infty}$, of LTHM when $\bar{d}\to\infty$ are given by$\colon$


\begin{align}
\label{eq:LTHM_PowerInfinitedelay}
\hspace{-6.5cm}\bar{P}_{\infty}\triangleq\underset{\bar{d}\to\infty}{\lim}\frac{1}{\bar{d}}\sum\limits_{m=1}^{J}\mathrm{E}_{\bar{Z}_m,\hat{Z}_m}\Big[\text{min}\left\{\bar{Z}_m,\hat{Z}_m\right\}\Big]P(\hat{h}_m,m),
\end{align}
\begin{align}
\label{eq:LTHM_DistortionInfinitedelay}
\hspace{-0.35cm}\bar{D}_{\infty}&\triangleq\lim_{\bar{d}\to\infty}\frac{1}{\bar{d}}\sum\limits_{m=1}^{J}\left\{\mathrm{E}_{\bar{Z}_m,\hat{Z}_m}\Big[[\bar{Z}_m-\hat{Z}_m]^{+}\Big]\sigma_m^2+\mathrm{E}_{\bar{Z}_m,\hat{Z}_m}\Big[\text{min}\left\{\bar{Z}_m,\hat{Z}_m\right\}\Big]D(\hat{h}_m,m)\right\},
\end{align}

\noindent where the allocated power $P(\hat{h}_m,m)$ and the distortion $D(\hat{h}_m,m)$ are chosen as in Eqn.~(\ref{eq:PowerAllocationLTHM}) and Eqn.~(\ref{eq:DistortionLTHM}), respectively.

In the rest of the proof, we use $p(m)$ to refer to the condition of Theorem~\ref{Theorem1}, i.e., $p_M(m)=p_H(\hat{h}_m)=p(m)$, $\forall m$. Under this condition, the expected value and variance of $\bar{Z}_m$ and $\hat{Z}_m$ can be found as, $\mathrm{E}[\bar{Z}_m]=\mathrm{E}[\hat{Z}_m]=\bar{d}\cdot p(m)$ and $\mathrm{Var}[\bar{Z}_m]=\mathrm{Var}[\hat{Z}_m]=\sigma_{{Z}_m}^2=\bar{d}\cdot p(m)\cdot(1-p(m))$, respectively. Let $\epsilon>0$ be any positive number. Then, the Chebyshev's inequality leads to the following inequalities, $\mathrm{Pr}\{ |\bar{Z}_m-\bar{d}\cdot p(m)| \ge\epsilon\cdot\sigma_{Z_m}\}\le\frac{1}{\epsilon^2}$ and $\mathrm{Pr}\{ |\hat{Z}_m-\bar{d}\cdot p(m)| \ge\epsilon\cdot\sigma_{Z_m}\}\le\frac{1}{\epsilon^2}$. We define the interval $\mathcal{I}$ on the real line as, $\mathcal{I}=[\bar{d}\cdot p(m)-\epsilon\cdot\sigma_{Z_m},\bar{d}\cdot p(m)+\epsilon\cdot\sigma_{Z_m}]$.

Next, we compute~(\ref{eq:LTHM_PowerInfinitedelay}) and~(\ref{eq:LTHM_DistortionInfinitedelay}) by finding upper and lower bounds on the expectation terms under the matching condition. Observe that,
\vspace{-0.2cm}
\small
\begin{align}
\label{eq:ExpBound1}
\lim_{\bar{d}\to\infty}\frac{1}{\bar{d}}\mathrm{E}_{\bar{Z}_m,\hat{Z}_m}\Big[\text{min}\left\{\bar{Z}_m,\hat{Z}_m\right\}\Big]\le&\lim_{\bar{d}\to\infty}\frac{1}{\bar{d}}\mathrm{E}_{\bar{Z}_m,\hat{Z}_m}\big[\bar{Z}_m\big]=p(m).
\end{align}
\normalsize

We can also lower bound this term as,

\vspace{-0.5cm}
\begin{align}
\label{eq:ExpLowerBound1}
&\lim_{\bar{d}\to\infty}\frac{1}{\bar{d}}\mathrm{E}_{\bar{Z}_m,\hat{Z}_m}\Big[\text{min}\left\{\bar{Z}_m,\hat{Z}_m\right\}\Big],\\
\label{eq:ExpLowerBound2}
&\ge\lim_{\bar{d}\to \infty}\frac{1}{\bar{d}}\mathrm{E}_{\bar{Z}_m,\hat{Z}_m}\Big[\text{min}\left\{\bar{Z}_m,\hat{Z}_m\right\}\Big|_{\begin{subarray}{l}\bar{Z}_m\in\mathcal{I},\\\hat{Z}_m\in\mathcal{I}\end{subarray}}\Big]\scriptstyle{\mathrm{Pr}\left\{\bar{Z}_m\in\mathcal{I},\hat{Z}_m\in\mathcal{I}\right\}},\\
\label{eq:ExpLowerBound3}
&\ge\lim_{\bar{d}\to\infty}\frac{1}{\bar{d}}\left(\bar{d}p(m)-\epsilon\sigma_{{Z}_m}\right)\left(1-\frac{1}{\epsilon^2}\right)^2,\\
\label{eq:ExpLowerBound4}
&=\lim_{\bar{d}\to\infty}\left(p(m)-\frac{\sqrt{p(m)(1-p(m))}}{\bar{d}^{\frac{1}{6}}}\right)\left(1-\frac{1}{\bar{d}^{\frac{2}{3}}}\right)^2=p(m),
\end{align}

\noindent where~(\ref{eq:ExpLowerBound2}) follows from the law of total expectation;~(\ref{eq:ExpLowerBound3}) follows from the definition of $\mathcal{I}$, and the Chebyshev's inequality; and~(\ref{eq:ExpLowerBound4}) is obtained by setting $\epsilon=\bar{d}^{\frac{1}{3}}$. Since the upper and lower bounds in~(\ref{eq:ExpBound1}) and~(\ref{eq:ExpLowerBound4}) are equal, we have shown that~(\ref{eq:ExpBound1}) converges to $p(m)$ as $\bar{d}\to\infty$.

Similarly,
\vspace{-0.2cm}
\begin{align}
\label{eq:ExpBound6}
&\lim_{\bar{d}\to\infty}\frac{1}{\bar{d}}\mathrm{E}_{\bar{Z}_m,\hat{Z}_m}\Big[[\bar{Z}_m-\hat{Z}_m]^{+}\Big],\\
\label{eq:ExpBound7}
=&\lim_{\bar{d}\to \infty}\frac{1}{\bar{d}}\Bigg\{\mathrm{E}_{\bar{Z}_m,\hat{Z}_m}\Big[[\bar{Z}_m-\hat{Z}_m]^{+}\Big|_{\begin{subarray}{l}\bar{Z}_m\in\mathcal{I},\\\hat{Z}_m\in\mathcal{I}\end{subarray}}\Big]\scriptstyle{\mathrm{Pr}\left\{\bar{Z}_m\in\mathcal{I},\hat{Z}_m\in\mathcal{I}\right\}}\nonumber\\
+&\mathrm{E}_{\bar{Z}_m,\hat{Z}_m}\Big[[\bar{Z}_m-\hat{Z}_m]^{+}\Big|_{\begin{subarray}{l}\bar{Z}_m\not\in\mathcal{I}\\\hspace{0.3cm}\text{or}\\ \hat{Z}_m\not\in\mathcal{I}\end{subarray}}\Big]\scriptstyle{\mathrm{Pr}\{\bar{Z}_m\not\in\mathcal{I}\hspace{0.1cm}\text{or}\hspace{0.1cm}\hat{Z}_m\not\in\mathcal{I}\}}\Bigg\},\\
\label{eq:ExpBound8}
\le&\underset{\bar{d}\to\infty}{\lim}\frac{1}{\bar{d}}\Bigg\{2\epsilon\sigma_{{Z}_m}+\left(\frac{2}{\epsilon^2}+\frac{1}{\epsilon^4}\right)\bar{d}\Bigg\},\\
\label{eq:ExpBound9}
=&\underset{\bar{d}\to\infty}{\lim}\Bigg\{\left(\frac{2\sqrt{p(m)(1-p(m))}}{\bar{d}^{\frac{1}{6}}}\right)+\left(\frac{2}{\bar{d}^{\frac{2}{3}}}+\frac{1}{\bar{d}^{\frac{4}{3}}}\right)\Bigg\}=0,
\end{align}

\noindent where~(\ref{eq:ExpBound7}) follows from the law of total expectation;~(\ref{eq:ExpBound8}) follows from the from the definition of $\mathcal{I}$, and the Chebyshev's inequality; and~(\ref{eq:ExpBound9}) is obtained by setting $\epsilon=\bar{d}^{\frac{1}{3}}$. This proves that~(\ref{eq:ExpBound6}) indeed converges to zero as $\bar{d}\to\infty$. This also implies that as $\bar{d}\to\infty$, all selected measurements by the LTSM strategy satisfy the hard matching condition. Hence, LTSM and LTHM are equivalent in the asymptotic of $\bar{d}\to\infty$ under the matching condition of Theorem~\ref{Theorem1}.

Finally, we can rewrite $\bar{P}_{\infty}$ and $\bar{D}_{\infty}$ for both LTHM and LTSM as$\colon$

\vspace{-0.8cm}
\begin{align}
\label{eq:appendix1}
\bar{P}_{\infty}&=\sum\limits_{m=1}^{J}\Big[{\mu^{\ast}}q-\frac{1}{|\hat{h}_m|^2}\Big]^+p(m),\\
\label{eq:appendix1-dist}
\bar{D}_{\infty}&=\sum\limits_{m=1}^{J}\Bigg[\frac{\sigma_{m}^2}{|\hat{h}_m|^2\Big[{\mu^{\ast}}q-\frac{1}{|\hat{h}_m|^2}\Big]^++1}\Bigg]p(m),
\end{align}

\noindent where we use $q\triangleq\frac{\sigma_m}{|\hat{h}_m|}$, $\forall m$, from Theorem~\ref{Theorem1}, and $\mu^{\ast}$ is chosen to satisfy $\bar{P}_{\infty}=P$.

Next, we show that $(\bar{P}_{\infty},\bar{D}_{\infty})$ pair above, obtained under the conditions of Theorem~\ref{Theorem1}, achieve the TLB pair $(\bar{P}_{e},\bar{D}_{e})$, derived in Section~\ref{s:TheoreticalLowerBound}. First, under the matching condition, observe that ${\mu^{\ast}}q=\alpha^*$, and thus, $\bar{P}_{\infty}=\bar{P}_{e}=P$. Moreover, under the matching condition, $\bar{R}_{e}=\bar{C}_e$ in TLB implies $\alpha^{\ast}=\frac{q^2}{\beta^{\ast}}$. Combining the two equalities, we obtain $\mu^{\ast}=\frac{q}{\beta^{\ast}}$. Substituting this into Eqn.~(\ref{eq:DistortionRate}) together with the matching condition, we can show that $\bar{D}_{e}=\sum\limits_{m=1}^{J}\min\left(\frac{q}{\mu^{\ast}},\sigma_{m}^2\right)p(m)=\bar{D}_{\infty}$, which concludes the proof of Theorem~\ref{Theorem1}.

\section{Proof of Lemma 3}
In order to prove Lemma~\ref{Lemma3}, we construct a counter-example. We argue that the achievable MSE distortion of a particular LT scheme that is not constrained to use only a one-to-one mapping between measurements and channels can be smaller than the minimum achievable MSE distortion of all possible LT schemes that use only a one-to-one mapping, i.e., a diagonal encoding matrix. Suppose we have $J=2$ zero-mean Gaussian parameters with variances $\sigma_1^2$ and $\sigma_2^2$, which are requested with probabilities $p_M(1)=p_1$ and $p_M(2)=p_2=(1-p_1)$, respectively, and assume an extreme case, where $\sigma_1^2>0$ and $\sigma_2^2=0$. Suppose we have a discrete fading channel with two states, which are observed with probabilities $p_{H_1}(\hat{h}_1)=p_1$ and $p_{H_2}(\hat{h}_2)=p_2$, respectively, and assume that the channel states are $\hat{h}_1>0$ and $\hat{h}_2=0$. We aim at linearly transmitting $N=2$ measurements of parameters $m_1\in[1\text{:}2]$ and $m_2\in[1\text{:}2]$, over $N=2$ channel states $h_1\in\{\hat{h}_1,\hat{h}_2\}$ and $h_2\in\{\hat{h}_1,\hat{h}_2\}$.

We first characterize the minimum achievable MSE distortion, $\bar{D}_1$, for all possible LT schemes with a diagonal encoding matrix. According to Eqn.~(\ref{eq:AveragePowerParallelChannels}), the encoding function needs to satisfy the average power constraint $P$, i.e., $\frac{1}{2}\left[P_{11}p_1^2+P_{12}p_1p_2+P_{21}p_1p_2+P_{22}p_2^2\right]=P$, where $P_{m_1m_2}$ is the allocated power for the pair of measurements of parameters $m_1$ and $m_2$, respectively. We have $P_{22}=0$, since $\sigma_2^2=0$. Then, by using Eqn.~(\ref{eq:AverageDistortionParallelChannels}), the MSE distortion $\bar{D}_1$ can be written explicitly as follows$\colon$

\vspace{-0.6cm}
\begin{align}
\label{eq:AverageDistortion-Diagonal1}
\bar{D}_1&=\frac{1}{2}\left\{p_1^2\left(\mathrm{E}_{H_1}\left[\frac{\sigma_{1}^2}{|h_1|^2\frac{P_{11}}{2}+1}\right]+\mathrm{E}_{H_2}\left[\frac{\sigma_{1}^2}{|h_2|^2\frac{P_{11}}{2}+1}\right]\right)
\right.
\nonumber \\
 & \left.
 +p_1p_2\left(\mathrm{E}_{H_1}\left[\frac{\sigma_{1}^2}{|h_1|^2P_{12}+1}\right]+\mathrm{E}_{H_2}\left[\frac{\sigma_{1}^2}{|h_2|^2P_{21}+1}\right]\right)\right\},\\
 \label{eq:AverageDistortion-Diagonal2}
 &=\textstyle{p_1^2\left(p_1\frac{\sigma_{1}^2}{|\hat{h}_1|^2\frac{P_{11}}{2}+1}+p_2\sigma_{1}^2\right)}+\textstyle{\frac{p_1p_2}{2}\left(p_1\frac{\sigma_{1}^2}{|\hat{h}_1|^2P_{12}+1}+p_1\frac{\sigma_{1}^2}{|\hat{h}_1|^2P_{21}+1}+2p_2\sigma_{1}^2\right)},
\end{align}

\noindent where the minimum distortion is achieved by dividing the power, i.e., $P_{11}$, equally between measurements if two measurements are observed from parameter $1$, i.e., $m_1=m_2=1$. If one measurement is requested from each parameter, i.e., $(m_1=1,m_2=2)$ or $(m_1=2,m_2=1)$, then the minimum distortion is achieved by allocating the entire power, i.e., $P_{12}$ or $P_{21}$, to the measurement of parameter $1$, since $\sigma_2^2=0$.

Assuming the average power constraint $P$ is satisfied as in the above scheme, we next consider a particular LT scheme. This scheme uses a diagonal encoding matrix if both measurements are observed from the same parameter; otherwise, it uses a non-diagonal matrix, where the measurement of parameter $1$ is transmitted over two channels. Then, from Eqn.~(\ref{eq:AverageDistortionParallelChannels}), the MSE distortion $\bar{D}_2$ can be written as follows$\colon$

\vspace{-0.6cm}
\begin{align}
\label{eq:AverageDistortion-NonDiagonal1}
&\bar{D}_2=\textstyle{\frac{1}{2}\left\{p_1^2\left(\mathrm{E}_{H_1}\left[\frac{\sigma_{1}^2}{|h_1|^2\frac{P_{11}}{2}+1}\right]+\mathrm{E}_{H_2}\left[\frac{\sigma_{1}^2}{|h_2|^2\frac{P_{11}}{2}+1}\right]\right)+p_1p_2\left(\mathrm{E}_{H_1,H_2}\left[\frac{\sigma_{1}^2}{(|h_1|^2+|h_2|^2)\frac{P_{12}}{2}+1}\right]\right.\right.}
\nonumber \\
 &\left.\textstyle{\left.
 +\mathrm{E}_{H_1,H_2}\left[\frac{\sigma_{1}^2}{(|h_1|^2+|h_2|^2)\frac{P_{21}}{2}+1}\right]\right)}\right\},\\
 \label{eq:AverageDistortion-NonDiagonal2}
 &=\textstyle{p_1^2\left(p_1\frac{\sigma_{1}^2}{|\hat{h}_1|^2\frac{P_{11}}{2}+1}+p_2\sigma_{1}^2\right)+\frac{p_1p_2}{2}\left(2p_2^2\sigma_{1}^2+p_1^2\frac{\sigma_{1}^2}{|\hat{h}_1|^2P_{12}+1}
  \right.}
 \nonumber \\
 &\textstyle{\left.
 +p_1^2\frac{\sigma_{1}^2}{|\hat{h}_1|^2P_{21}+1}+2p_1p_2\frac{\sigma_{1}^2}{|\hat{h}_1|^2\frac{P_{12}}{2}+1}+2p_1p_2\frac{\sigma_{1}^2}{|\hat{h}_1|^2\frac{P_{21}}{2}+1}\right)},
\end{align}

\noindent where the minimum distortion can be achieved by dividing the power, i.e., $P_{11}$, equally between measurements if two measurements are observed from parameter $1$, i.e., $m_1=m_2=1$, similarly to the above scheme. If one measurement is requested from each parameter, i.e., $(m_1=1,m_2=2)$ or $(m_1=2,m_2=1)$, then this particular scheme divides the power, i.e., $P_{12}$ or $P_{21}$, equally between two channels $h_1$ and $h_2$ for the transmission of the measurement of parameter $1$, as seen in the term multiplied by $p_1p_2$ in~(\ref{eq:AverageDistortion-NonDiagonal1}). If two measurements are observed from parameter $2$, i.e., $m_1=m_2=2$, then we do not allocate power, i.e., $P_{22}=0$, since $\sigma_2^2=0$.

We can easily see that $\bar{D}_2<\bar{D}_1$ for all $P_{11}$, $P_{12}$ and $P_{21}$. This implies that the minimum achievable MSE distortion of LT schemes constrained to one-to-one mapping can be improved by utilizing non-diagonal encoding matrices, which concludes the proof of Lemma~\ref{Lemma3}.

\bibliographystyle{MyIEEEtran}
\bibliography{TWC_arxiv}

\end{document}